\documentclass[proc]{edpsmath}

\usepackage[T1]{fontenc}
\usepackage{todonotes}
\usepackage{amsmath, amssymb, dsfont}
\usepackage{comment}
\usepackage{braket}
\usepackage{qcircuit}

\usepackage{caption}
\usepackage{float, tabularx}
\usepackage{hyperref}

\ifpdf
\hypersetup{
  pdftitle={A Quantum Spectral Framework for Solving PDEs},
  pdfauthor={C.-K. Huang, G. Antonioli and F. Barbaresco}
}
\fi
\theoremstyle{definition}

\begin{document}

\title{A Quantum Spectral Framework for Solving PDEs}
\author{Chih-Kang Huang}\address{Université Côte d'Azur, Inria, CNRS, LJAD, Nice, France; \email{chih-kang.huang@inria.fr}}
\author{Giacomo Antonioli}\address{Università di Pisa, Pisa, Italy; \email{giacomo.antonioli@phd.unipi.it}}
\author{Frédéric Barbaresco}\address{Thales Land \& Air Systems; \email{frederic.barbaresco@thalesgroup.com}}

%
%
\begin{abstract} 
Partial differential equations (PDEs) are fundamental across numerous scientific fields. As these problems scale to high dimensions, classical numerical schemes introduce severe computational bottlenecks, known as the curse of dimensionality. Attempts to solve this problem typically rely on either classical sparsity and low-rank decompositions, or neural network surrogate models. On the other hand, Quantum Computing offers a promising alternative, as it allows us to operate in significantly larger spaces while demanding far fewer resources. In this work, we present a quantum subroutine to solve second-order linear PDEs by exploiting the structural properties of the filter in Fourier space using Quantum Block Encoding (QBE) with quantum reversible arithmetic. This approach serves as a specialized alternative to standard quantum matrix inversion, which typically relies solely on Quantum Singular Value Transformation (QSVT) without exploiting the inherent structural properties of the matrix. We validate the proposed methodology against its classical counterpart to prove its correctness.
This framework provides a foundation for extending these methods toward quantum group Fourier transforms, wavelet-based analysis, and equivariant quantum neural networks (EQNNs), offering a promising path toward solving broader classes of problems, including nonlinear PDEs.
\end{abstract}
%
%
%
\maketitle

\section{Introduction}

Partial differential equations (PDEs) are foundational to virtually every field of science and engineering, modelling phenomena ranging from fluid dynamics and electromagnetism to materials science and financial markets. The development of efficient numerical solvers for PDEs is a cornerstone of scientific computing. Classical methods like the finite difference method and spectral methods are highly mature but can face significant computational bottlenecks, especially for high-dimensional problems where these methods scale poorly with the dimensionality. For instance, consider the Poisson equation defined on a domain \(\Omega \subset \mathbb{R}^d\):
\begin{equation}
  \Delta u = f,
  \label{eq:intro_poisson}
\end{equation}
the finite difference and spectral methods give time complexities of \(\mathcal{O}(\operatorname{poly}(N^d))\) and \(\mathcal{O}(\operatorname{poly}((\log N)^d))\) respectively, with \(N\) the number of grid points per dimension, see~\cite{childs2021high}. Since both complexities scale exponentially with the spatial dimension \(d\), these methods suffer from the \emph{curse of dimensionality}.

Several techniques have been developed to tackle this issue. For example, classical approaches such as sparse grid discretization~\cite{bungartz2004sparse} and low-rank tensor decompositions~\cite{khoromskij2012tensors} attempt to mitigate the exponential growth of degrees of freedom by exploiting the structural regularity of the solution. While these methods can be effective for specific classes of smooth functions, they often require strong structural assumptions to remain efficient in high dimensions. More recently, other approaches incorporate neural networks as \emph{surrogate models} to learn PDE solutions. While several works demonstrate that the size of such networks scales polynomially with the dimension for semi-linear parabolic PDEs~\cite{de2024numerical, furuya2024quantitative}, the complex nonlinearities involved in the training process often result in a lack of convergence theory and of rigorous error analysis.

Quantum algorithms offer the potential to overcome these limitations by exploiting the structure of linear operators through unitary embeddings and quantum linear algebra techniques~\cite{harrow2009quantum, jin2023quantum}. A key paradigm enabling such approaches is Quantum Block Encoding (QBE)~\cite{gilyen2019quantum}, which provides a framework for representing non-unitary operators as leading blocks of unitary matrices. This framework is used by a wide range of quantum algorithms, including quantum linear system solvers~\cite{harrow2009quantum, childs2021high} and Hamiltonian simulation~\cite{low2019hamiltonian}. Consequently, the efficient construction of block encodings for structured operators is a fundamental problem in quantum algorithm design.

A central challenge in this construction is the implementation of the inverse spectral filter, which requires encoding functions of the form $1/x$. Standard approaches rely on Quantum Singular Value Transformation (QSVT) to approximate such functions~\cite{gilyen2019quantum,lubasch2025quantum}. In contrast, we exploit a block-encoding method for diagonal matrices where inversion is implemented via reversible arithmetic circuits~\cite{tong2021fast, antonioli2026quantumblockencodingsemiseparable}.
In this work, we propose a quantum framework, that exploits block encoding, to solve second-order elliptic PDEs with constant coefficients and their parabolic counterparts, discretized on a grid of \(N = 2^n\) points per dimension. In this setting, the differential operator becomes diagonal in the Fourier basis.

Our modular construction combines: (i) a unitary Quantum Fourier Transform (QFT)~\cite{preskill}, (ii) an efficient block encoding of the diagonal filter via reversible arithmetic, and (iii) composition into the full solution operator. The subroutine accepts the source term as a quantum state $\ket{f}$ and returns $\ket{u}$ assuming that the initial is already prepared in the quantum register, and the availability of efficient state preparation techniques, such as Quantum RAM~\cite{qram,kerenidis}.

We further demonstrate how this block encoding can be applied to some elliptic PDEs such as Poisson equation, Helmholtz equation, and parabolic diffusion equations. 
In the diffusion case, one can approximate the solution by a time-stepping implicit scheme and spectral method using the same Fourier-diagonal structure, one can then iterate the identical circuit architecture at each step. However, maintaining the solution state entirely within the quantum register is necessary to preserve speedup; extracting the state classically via measurements after each step incurs exponential cost. Extending the framework to a fully quantum iterative solver is left for future work.

We validate the construction by comparing the logical action of the block-encoding unitary, simulated in a noiseless quantum setting under the assumption of exact arithmetic, against classical spectral solvers. These quantum simulations confirm that the composition of the QFT and the diagonal block-encoded filter faithfully reproduces the expected solution operator, thereby verifying the correctness of the block-encoding framework. This work highlights how the structural property of the differential operators can be systematically leveraged within the block encoding framework. More broadly, it provides a structure-exploiting, modular alternative to general quantum linear system solvers.

The paper is organized as follows. 
In Section~\ref{sec:spectral_method}, we review the application of spectral methods to elliptic and parabolic PDEs with constant coefficients and periodic boundary conditions, we focus in particular on the Kronecker-based spectral method which is naturally adapted for the quantum implementation. 
In Section~\ref{sec:quantum_block_encoding}, we introduce our Block Encoding framework and describe how diagonal matrices are efficiently encoded and inverted inside a quantum circuit. 
In Section~\ref{sec:numerical_experiments}, we provide an in-depth comparison between the classical spectral method and its quantum counterparts. 
Finally, we conclude in Section~\ref{sec:conclusions}.

\section{Spectral method for PDEs with periodic conditions}
\label{sec:spectral_method}

In this section, we present a spectral method based on \emph{Kronecker products} for solving PDEs with periodic boundary conditions. While this approach is less computationally efficient than its non-Kronecker counterparts based on Fast Fourier Transform algorithms, we still adopt it in this work in order to compare it directly with the quantum implementations. Indeed, in the quantum setting, the Kronecker formulation naturally represent the multi-dimensional differential operators due to the tensor product structure of multi-qubit systems, see~\cite{steeb2011matrix, harrow2009quantum}.

The spectral method is an efficient numerical technique for solving partial differential equations endowed with periodic boundary conditions. Throughout the rest of the paper, we consider second-ordered PDEs defined on a $d$-dimensional flat torus $\mathbb{T}^d$, discretized on a grid of $N=2^n$ points per dimension for $n \in \mathbb{N}$. 
The spectral method aims to approximate solutions by truncated Fourier series, see~\cite{boyd2001chebyshev} for instance.

Compared to classical grid-based schemes such as finite-difference and finite-element methods whose convergence rates are polynomial, the convergence rate of the spectral method can achieve exponential convergence, see~\cite{gottlieb1977numerical}. For instance, consider the Poisson equation on the $d$-dimensional flat torus: 
\begin{equation}
\Delta u = f \quad \text{on} \quad \mathbb{T}^d,
\label{eq:poissons_eq}
\end{equation}
it is known that the error between the spectral approximation $u_{\text{spectral}}$ and the analytical solution $u$ is given by
\begin{equation}
  \| u - u_{\text{spectral}} \|_2  = O(e^{-cN})
\label{eq:spectral_convergence}
\end{equation}
where $c > 0$ is a constant only depending on $f$ and $\| \cdot \|_2$ is the $L^2$-norm on $\mathbb{T}^d$, whereas the convergence rate of finite-difference method is $O(N^{-2})$. For a detailed comparison of different numerical methods, see Chapter 4 in~\cite{trefethen2000spectral}, for instance.
In the following, we introduce the Kronecker-based spectral method for elliptic and parabolic PDEs, which serves as the blueprint for the quantum circuits implemented in the subsequent section.
We refer to~\cite{evans2022partial} for a complete review of elliptic and parabolic PDEs.

\subsection{Spectral Methods for Elliptic PDEs with constant coefficients}
\label{subsec:elliptic_pdes}

We focus on elliptic PDEs with \emph{constant coefficients} of the form:
\begin{equation}
\nabla \cdot (A \nabla u) = \sum_{i, j=1}^d A_{ij} \partial_{ij}^2 u = f\quad \text{on} \quad \mathbb{T}^d 
\label{eq:elliptic_PDE_with_constant_coeffs}
\end{equation}
where $f\in L^2(\mathbb{T}^d)$ satisfies the zero-mean condition 
\[
\int_{\mathbb{T}^d} f(x) dx = 0,
\]
$\partial_{ij}^2 u := \frac{\partial^2 u}{\partial x_i \partial x_j}$ for every $i, j \in \{1, 2, \ldots, d\}$, and $A \in \mathcal{M}_d(\mathbb{R})$ is a symmetric positive-definite matrix, that is, all the eigenvalues of $A$ are strictly positive. 

These conditions ensure that the differential operator is strictly elliptic, and thus that the solutions to Eq.~\eqref{eq:elliptic_PDE_with_constant_coeffs} are smooth across the domain. In particular, the Poisson equation~\eqref{eq:poissons_eq} corresponds to the case where $A = \mathrm{I}_{d}$ is an identity matrix.
The solutions to \eqref{eq:elliptic_PDE_with_constant_coeffs} are given by:
\begin{equation}
u = \mathcal{G} * f + c \label{eq:elliptic_sols}
\end{equation}
for any $c \in \mathbb{R}$, where $\mathcal{G}$ is the \emph{Green's function} of \eqref{eq:elliptic_PDE_with_constant_coeffs}, see~\cite{taylor1996partial, aubin1998some}. 
On the periodic torus $\mathbb{T}^d$, $\mathcal{G}$ is defined by :
\begin{equation}
\nabla \cdot( A \nabla \mathcal{G}) = \delta_0 - \frac{1}{\mathrm{Vol}(\mathbb{T}^d)}. \label{eq:greens_sol}
\end{equation}

Instead of obtaining the solutions \eqref{eq:elliptic_sols} by computing the convolution of the Green's function $\mathcal{G}$ with the source term $f$ in physical space, the spectral method consists of computing the convolution in Fourier space \cite{trefethen2000spectral}. Applying the Fourier series to \eqref{eq:elliptic_sols}, we get:
\begin{equation}
\hat u = \hat{\mathcal{G}}  \odot \hat f + \hat c
\end{equation}
where $\hat u$, $\hat{\mathcal{G}}$, $\hat f$ and $\hat c$ are the Fourier coefficients of $u, \mathcal{G}, f$ and $c$ respectively, $\hat{\mathcal{G}}$ is referred to as the \emph{spectral filter}, and $\odot$ is the Hadamard element-wise product of matrices. 
Thanks to \eqref{eq:greens_sol}, we construct the spectral filter using Kronecker products as:

\begin{equation}
\hat{\mathcal{G}}_{\text{elliptic}} = \left( \sum_{i,j=1}^d A_{ij} \mathbf{D}_i \mathbf{D}_j +  E_{11}\right)^{-1} \label{eq:elliptic_spectral_filter}
\end{equation}
where $\mathbf{D}_i$ is a matrix of size $N^d \times N^d$ defined by
\begin{equation}
\mathbf{D}_i = I_N \otimes \dots \otimes \underbrace{D_N}_{i\text{-th slot}} \otimes \dots \otimes I_N
\label{eq:bf_D}
\end{equation}
with $D_N = 2\pi i \times \mathrm{diag}(0, 1, 2, \ldots, N-1)$ and $E_{11} = (\delta_{i1} \delta_{1j})_{i, j = 1}^d \in \mathcal{M}_d(\mathbb{R})$ the elementary matrix. Notice that the additional term $E_{11}$ in Eq.~\eqref{eq:elliptic_spectral_filter} is to ensure that the operator is invertible on the zero-frequency subspace.

\subsubsection{Helmholtz equation}
\label{par:Helm}
The Helmholtz equation arises naturally when seeking time-harmonic solutions to the wave equation. This equation is fundamental for describing steady-state phenomena in numerous fields such as acoustics, seismology, and electromagnetic radiation~\cite{sommerfeld1949partial, etgen2009overview}. In these contexts, the equation describes how waves of a specific frequency scatter and resonate within a medium, see~\cite{colton2013integral}. 

Similarly to the elliptic case,  here we consider the Helmholtz equation on the torus:
\begin{equation}
\Delta u + \lambda^2 u = f\quad \text{on} \quad \mathbb{T}^d \label{eq:helmholtz_equation}
\end{equation}
where $\lambda$ is a positive real number and $f$ is a given source term. The corresponding spectral filter is given by:
\begin{equation}
  \hat{\mathcal{G}}_{\text{Helmholtz}} = \left (\sum_{i=1}^d \mathbf{D}_i^2 + \lambda^2 I_{N^d}\right)^{-1} 
  \label{eq:helmholtz_filter}
\end{equation}
where $\mathbf{D}_i$ is defined as in Eq.~\eqref{eq:bf_D}
with $D_N = 2\pi i \times \mathrm{diag}(0, 1, 2, \ldots, N-1)$. Note that the operator is invertible as long as $-\lambda^2$ does not coincide with an eigenvalue of the Laplacian on the torus.
Unlike the Poisson equation~\eqref{eq:poissons_eq}, the $\lambda^2$ term typically ensures that the Helmholtz operator is invertible even on the zero-frequency mode, since $\lambda\neq 0$.

\subsection{Spectral Methods for Parabolic PDEs with constant diffusion coefficients}
\label{subsec:parabolic_pdes}

Parabolic PDEs model time-dependent phenomena characterized by smoothing and dissipative processes, such as heat conduction and chemical diffusion \cite{evans2022partial}. In this section, we consider parabolic equations with \emph{constant diffusion coefficients}, which are the parabolic versions of the elliptic PDEs discussed in Section~\ref{subsec:elliptic_pdes}.

We consider the \emph{anisotropic diffusion equation} as a parabolic PDE with constant coefficients given by:
\begin{equation}
\frac{\partial u}{\partial t} = \nabla \cdot ( A \nabla u)(x,t) - f(x,t) \quad \text{on} \quad \mathbb{T}^d \times (0, T] \label{eq:diffusion_eq}
\end{equation}
for some $T>0$, endowed with periodic boundary condition and with  $A\in \mathcal{M}_d(\mathbb{R})$ a symmetric positive definite matrix and $f$ a source term.
In particular, when $A$ is the identity matrix, we recover the (isotropic) heat equation:
\begin{equation}
\frac{\partial u}{\partial t} = \Delta u(x,t) - f(x,t) \quad \text{on} \quad \mathbb{T}^d \times (0, T] \label{eq:heat_eq}
\end{equation}

In this work, we focus on the case where the source term is \emph{time-independent}, that is, $f(x, t) = f(x)$ for every $t\in (0, T]$. Under this assumption, Eq.~\eqref{eq:diffusion_eq} can be written as a $L^2$-gradient flow. Namely, the equation is an energy-minimizing flow given by 
\begin{equation}
  \frac{\partial u}{\partial t} = -\nabla_{L^2} E(u)
\label{eq:gradient_flow}
\end{equation}
where $E$ is the energy functional defined by 
\begin{equation}
E(u)= \int_{\mathbb{T}^d} \left( \frac{1}{2} A \nabla u \cdot \nabla u + f u \right) dx,
\label{eq:energy_functional}
\end{equation}
and $\nabla_{L^2} E(u)$ is the Gateaux derivative in $L^2(\mathbb{T}^d)$.

It is generally impractical to find exact analytical solutions for Eq.~\eqref{eq:diffusion_eq}, thereby requiring numerical approximations. We treat the temporal and spatial variables separately, and use an iterative \emph{implicit Euler scheme} that combines time discretization with the spectral method, see~\cite{quarteroni2006spectral}. 

The choice of an implicit formulation here is for its numerical stability. Explicit schemes are subject to a CFL condition ($\Delta t \lesssim \mathcal{O}(N^{-2})$) where $N$ is the number of grid points per dimension, due to the quadratic scaling of spectral eigenvalues, which often leads to numerical instability, see~\cite{strikwerda2004finite}. In contrast, the implicit scheme is \emph{unconditionally stable}, which allows to take significantly larger time steps compared to the explicit scheme, see~\cite{hesthaven2007spectral}.

At each iteration $n$, the solution for Eq.~\eqref{eq:diffusion_eq} given by the implicit scheme writes 
\begin{equation}
u_{n+1} = \left ( I - \Delta t \nabla \cdot (A \nabla) \right )^{-1}( u_n - \Delta t f).
\label{eq:implicit_scheme}
\end{equation}
Applying the Fourier transform to Eq.~\eqref{eq:implicit_scheme}, we obtain: 
\begin{equation}
  \hat{u}_{n+1} = \hat{\mathcal{G}}_{\text{diffusion}} \odot (\hat{u}_n - \Delta t \hat{f})
\label{eq:implicit_scheme_spectral}
\end{equation}
where $\hat{\mathcal{G}}_{\text{diffusion}}$ is the corresponding spectral filter defined by 
\begin{equation}
  \hat{\mathcal{G}}_{\text{diffusion}} = \left (I_{N^d} - \Delta t \sum_{i, j = 1}^d A_{ij} \mathbf{D}_i \mathbf{D}_j \right)^{-1} 
  \label{eq:diffusion_filter}
\end{equation}
where $\mathbf{D}_i$ is defined as in Eq.~\eqref{eq:bf_D} in Section~\ref{subsec:elliptic_pdes} with $D_N = 2\pi i \times \mathrm{diag}(0, 1, 2, \ldots, N-1)$.

By utilizing the Kronecker structure, the spatial operator in the inverse remains diagonal in the Fourier basis. This diagonal structure allows for an efficient mapping to quantum diagonal unitary gates, avoiding the overhead of fully inverting a dense matrix, as detailed in Section~\ref{sec:numerical_experiments}.

\section{Quantum Block Encoding}
\label{sec:quantum_block_encoding}

In this section, we detail the quantum circuits required to implement the spectral filters $\hat{\mathcal{G}}$ derived in Section~\ref{sec:spectral_method}.
Rather than treating the PDE solver as a stand-alone process, we integrate it with the
QBE framework to create a modular quantum subroutine, providing a building block for larger quantum workflows. A key advantage of this approach is that whenever a problem already provides the lattice as a quantum state, the solution is obtained
by simply applying the subroutine, with no additional state preparation
overhead. Consequently, the elliptic solver and the spectral methods
discussed here can be seamlessly embedded into broader quantum algorithms
wherever the resolution of a PDE is required as a mid-circuit step.

\subsection{The Block Encoding Framework}

Many quantum algorithms, such as those for linear systems or Hamiltonian
simulation, frequently require the application of non-unitary and
non-Hermitian matrices. Since quantum evolution is strictly unitary,
such matrices cannot be implemented directly. Block encoding overcomes this limitation
by embedding the matrix of interest into a larger Hilbert space where a global unitary exists, and it therefore provides the
natural framework for representing non-unitary spectral filters within
unitary operators.

The construction proceeds as follows. If $A$ acts on $n$ qubits, one
introduces $a$ ancilla qubits to form an extended system of $n + a$
qubits. Denoting the rescaled matrix as $\tilde{A} = A/\alpha$, the
block encoding $U_{\tilde{A}}$ is an $(n+a)$-qubit unitary whose
top-left block recovers $\tilde{A}$:
\begin{equation}
    U_{\tilde{A}} = \begin{pmatrix} \tilde{A} & * \\ *&* \end{pmatrix}.
\end{equation}
Applying $U_{\tilde{A}}$ to a state in which the ancilla register is
initialized to $\ket{0}^{\otimes a}$ yields the desired transformation
$\tilde{A}$ on the target register, conditioned on the ancilla being
measured in $\ket{0}^{\otimes a}$. Before stating the formal definition, we recall that $\|\cdot\|_2$
denotes the \textit{spectral norm} (operator $2$-norm), defined as the
largest singular value of a matrix:
\begin{equation}
    \|A\|_2 
    \;=\; 
    \sup_{\|\mathbf{x}\|_2 = 1} \|A\mathbf{x}\|_2
    \;=\;
    \sigma_{\max}(A),
\end{equation}
where $\|\mathbf{x}\|_2 = \bigl(\sum_i |x_i|^2\bigr)^{1/2}$ is the
Euclidean vector norm. Since any unitary $U_{\tilde{A}}$ satisfies
$\|U_{\tilde{A}}\|_2 = 1$, every principal block must have spectral
norm at most one. The rescaled matrix $\tilde{A}$ can therefore appear
as such a block only if $\|A/\alpha\|_2 \le 1$, which is equivalent
to requiring $\alpha \ge \|A\|_2$~\cite{camps2024explicit}. In this
context, $\alpha$ acts as a normalization factor that brings $A$ within
the unit spectral ball.

\begin{dfntn}[Block Encoding]
Let $A$ be an $n$-qubit operator and $\tilde{A} = A/\alpha$. We say
that an $(n+a)$-qubit unitary $U_{\tilde{A}}$ is an
$(\alpha, a, \epsilon)$-block encoding of $A$ if it satisfies:
\begin{equation}
    \left\| A - \alpha(\bra{0^a} \otimes I_n) U_{\tilde{A}}
    (\ket{0^a} \otimes I_n) \right\|_2 \le \epsilon,
\label{eq:block_encoding_def}
\end{equation}
where $\alpha \ge \|A\|_2$ is the normalization factor, $a$ is the
number of ancilla qubits, and $\epsilon$ accounts for the numerical
approximation error of the encoding.
\end{dfntn}

To see how the encoded matrix is extracted in practice, consider the system initialized in $\ket{0^a} \otimes \ket{\psi}$, where $\ket{\psi}$ is the input state on the $n$ target qubits. After applying $U_{\tilde{A}}$, post-selecting the ancilla register onto $\ket{0^a}$ yields the output state $\tilde{A}\ket{\psi}/\|\tilde{A}\ket{\psi}\|_2$, with success probability $\|\tilde{A}\ket{\psi}\|_2^2$.

\subsection{Specialization to Diagonal Block Encodings}

The spectral methods discussed in Section~\ref{sec:spectral_method} rely fundamentally on the Quantum Fourier Transform to diagonalize linear differential operators. As shown in the formulations for the Poisson equation, the elliptic equation with constant coefficients, and the heat equation, the core computational step involves the application of the spectral filter.

A crucial observation is that in the Fourier domain, all these operators appear as diagonal matrices. Their entries are explicit, efficiently computable functions of the frequency modes, specifically those derived from the eigenvalues of the discretized differential operators $D_N$ via the Kronecker-product formulations in \eqref{eq:elliptic_spectral_filter} and \eqref{eq:helmholtz_equation}.

Consequently, the quantum implementation of these solvers reduces to the task of applying the spectral filter $\hat{\mathcal{G}}$ as a diagonal matrix to the quantum state representing the solution. Since $\hat{\mathcal{G}}$ is diagonal in the Fourier basis, this operation can be executed efficiently using a sequence of diagonal quantum gates.

This necessitates a framework 
capable of not only block-encoding a diagonal matrix defined by an 
arithmetic function but also efficiently encoding its inverse. 
By exploiting this diagonal structure, we can bypass the significant resource overhead associated with generic matrix encoding protocols. Specifically, methods such as FABLE \cite{camps2022fable} typically require $O(N^2)$ gates for a general $N \times N$ matrix, while Linear Combination of Unitaries (LCU) approaches \cite{chakraborty2024implementing} scale with the sparsity or require $O(N)$ resources for dense representations. In contrast, our specialized arithmetic-based approach utilizes reversible circuits to load filter coefficients and their reciprocals with a gate complexity of only $O(\text{polylog}(N))$, making it more efficient for 
high-resolution PDE simulations.

\subsection{Implementation of Diagonal Block Encodings}
To formalize our approach, we consider a generic diagonal matrix $D \in \mathbb{R}^{N \times N}$, acting on the computational basis states of a $n$-qubit register (where $N=2^n$). The action of $D$ is defined by:
\begin{equation}
    D \ket{j} = d_j \ket{j}, \quad \text{for } j = 0, \dots, N-1,
\end{equation}
where the coefficients $d_j \in \mathbb{R}$ represent the eigenvalues of the operator. In the context of the  spectral PDE solvers introduced in Section~\ref{sec:spectral_method}, these values $d_j$ correspond directly to the evaluation of the filter function or differential operator at the specific Fourier mode $j$.
To implement the operator $U_D$ and its inverse efficiently, we rely on the conversion of classical arithmetic functions into reversible quantum circuits using ancillary qubits, a technique grounded in Lemma 10.10 of \cite{arora2009computational} and further detailed in \cite{lin2022lecture, li2023efficient}.
This process requires mapping a computed classical value, stored in the QRAM, into the amplitude of a target qubit, which we achieve via a specialized binary representation of angles.
 We begin by formally defining the quantum representation of a signed angle.

\begin{dfntn}[Binary Angle Representation]
Let $\vartheta \in (-1, 1)$ be a signed value. We define its $(t+1)$-qubit binary representation state $\ket{\vartheta}$ as:
\begin{equation}
    \ket{\vartheta} = \ket{s} \otimes \ket{\theta_{t-1}}\ket{\theta_{t-2}}\dots\ket{\theta_0},
\end{equation}
where $s \in \{0, 1\}$ represents the sign (with $s=0$ for positive and $s=1$ for negative), and $\theta_k \in \{0, 1\}$ are the bits of the fixed-point binary expansion of the magnitude $|\vartheta| = \sum_{k=0}^{t-1} \theta_k 2^{-(t-k)}$.
\end{dfntn}

To utilize this representation for the block encoding of our filter
operators, we assume the existence of a classical Boolean circuit with
depth $O(\text{polylog}(N))$ that computes a $(t+1)$-digit fixed-point
approximation of the function $\frac{1}{\pi} \arcsin(x)$, which can be implemented in a quantum circuit with the same complexity. Here, the input 
$x$ is restricted to the interval $[0, c]$, where $c = \frac{\sqrt{\pi^2-1}}
{\pi}$ is a constant chosen to bound the derivative of the arcsine. This 
restriction ensures that the precision $\epsilon$ can be maintained 
without an exponential increase in the number of bits $t$ as the argument 
approaches the unit boundary.

\begin{lmm}
\label{lemma:controlled_rotation}
There exists a unitary operator $U_\theta$, acting on the composite system of the angle register $\ket{\vartheta}$ and a single ancillary qubit initialized in $\ket{0}$, such that:
\begin{equation}
    U_\theta \ket{\vartheta}\ket{0} = \ket{\vartheta} (-1)^s \left( \cos(\pi \vartheta) \ket{0} + \sin(\pi \vartheta) \ket{1} \right).
\end{equation}
\end{lmm}

The quantum circuit sequence realizing $U_\theta$ through controlled $R_y$ rotations is illustrated in Figure~\ref{fig:controlled_rotation_circuit}.

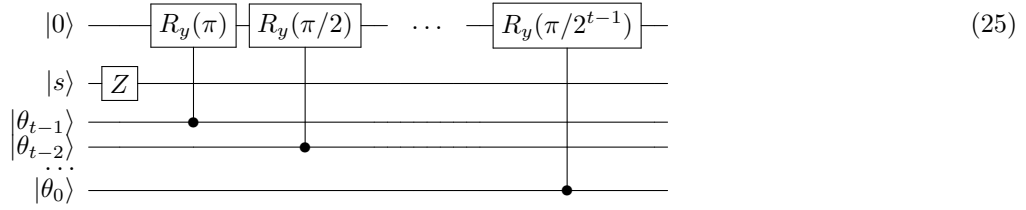
\begin{figure}[htpb]
    \centering
    \begin{align}   
\Qcircuit @C=0.5em @R=.7em {
\lstick{\ket{0}}            &\qw & \gate{R_y(\pi)} & \gate{R_y(\pi/2)} & \qw &\qw &     &     & \ldots &     &     &     & \qw & \gate{R_y(\pi/2^{t-1})} & \qw & \qw \\
\lstick{\ket{s}}            &  \gate{Z}  & \qw               & \qw                 & \qw &\qw & \qw & \qw & \qw    & \qw & \qw & \qw & \qw & \qw                       & \qw        & \qw \\
\lstick{\ket{\theta_{t-1}}} & \qw       & \ctrl{-2}         & \qw                 & \qw &\qw & \qw & \qw & \qw    & \qw & \qw & \qw & \qw & \qw                       & \qw        & \qw \\
\lstick{\ket{\theta_{t-2}}} & \qw       & \qw               & \ctrl{-3}           & \qw &\qw & \qw & \qw & \qw    & \qw & \qw & \qw & \qw & \qw                       & \qw        & \qw \\
\lstick{\dots}  \\
\lstick{\ket{\theta_0}}     & \qw       & \qw               & \qw                 & \qw &\qw & \qw & \qw & \qw    & \qw & \qw & \qw & \qw & \ctrl{-5}                 & \qw         & \qw
}
    \end{align}
    \caption{Quantum circuit for controlled rotations from binary angle representations. The circuit uses angle qubits to control a sequence of rotation gates with exponentially decreasing rotation angles, implementing the unitary $U_\theta$ that converts fixed-point binary values into quantum rotation angles.}
    \label{fig:controlled_rotation_circuit}
\end{figure}

\subsubsection{Block Encoding the Inverse of the Diagonal Filter}
To implement the block encoding of the inverse filter $D^{-1}$, we exploit the 
diagonal structure where $(D^{-1})_{ii} = 1/D_{ii}$. To ensure that $D$ is 
invertible, we utilize the regularized or shifted frequency modes defined in 
Section~\ref{subsec:elliptic_pdes}. The physical implementability and manageable 
resource overhead of this operation are theoretically supported by the 
\textit{fast inversion} primitive introduced by Tong et al.~\cite{tong2021fast}. 
Their work rigorously demonstrates that explicitly block-encoding a matrix inverse 
by implementing eigenvalue inversion via classical arithmetic within a quantum circuit 
is highly efficient, requiring a non-Clifford gate depth that scales merely as 
$O(\operatorname{polylog}(1/\varepsilon))$. By avoiding the deep circuit scaling 
associated with generic QSVT approaches, this foundation guarantees that our arithmetic 
subroutines will not destroy the theoretical quantum speedup. We formally state our 
inversion mechanism below based on~\cite{antonioli2026quantumblockencodingsemiseparable,tong2021fast}; for a more in-depth explanation and the complete derivation, 
we refer the reader to Ref.~\cite{antonioli2026quantumblockencodingsemiseparable}.

\begin{thrm}[Block encoding of the inverse diagonal filter]\label{theo:invBE}
Let $D \in \mathbb{R}^{N \times N}$ be an invertible diagonal matrix with
$M = \max_{i}|D_{ii}|$, $m = \min_{i}|D_{ii}|$, and condition number
$\kappa := M/m$.  Set
$$
  t  = \left\lceil \log_2\!\frac{\pi M}{\varepsilon} \right\rceil
     = O\!\left(\log\frac{M}{\varepsilon}\right)
  \quad \text{and} \quad
  t' = t + \lceil\log_2 \kappa\rceil
     = O\!\left(\log\frac{\kappa M}{\varepsilon}\right).$$
     
Assume access to an oracle $O_D$ that writes a $(t'+1)$-bit binary
representation of $|\tilde D_{ii}|/M$ into an ancilla register
(so the oracle precision scales as $O(\log(\kappa M/\varepsilon))$ bits).
Let $c = \sqrt{\pi^2-1}/\pi$ and $k = c/\kappa$.
Assume there exists a classical Boolean circuit of depth
$O(\operatorname{polylog}(N))$ computing a $t$-bit fixed-point
approximation of $x\mapsto\tfrac{1}{\pi}\arcsin(k/x)$ on $[1/\kappa,1]$, that can be implemented as a quantum oracle B with the same complexity.

Then a $(c/m,\,1,\,\varepsilon)$-block encoding of $D^{-1}$ can be
realized by a quantum circuit with gate depth
$$
  O\!\left(\operatorname{polylog}(N) + \log\frac{\kappa M}{\varepsilon}\right),
$$
where the second term captures the additive contributions of the
condition number and the target precision plus two calls to the oracle $O_D$ and two calls to the oracle $B$.
Although the circuit internally allocates
$O\!\bigl(\log(\kappa M/\varepsilon)\bigr)$ ancilla qubits for the
$(t'+1)$-bit oracle register and the $t$-bit arithmetic register,
both are returned to $\ket{0}$ by the uncomputation steps $B^\dagger$
and $O_D^\dagger$ and do not contribute to the final ancilla footprint.
The resulting block encoding therefore requires only a single
ancilla qubit, namely the one consumed by the controlled rotation
$U_\theta$.
\end{thrm}

\begin{proof}[Proof sketch]
The circuit implementation follows the fast inversion framework of 
Tong et al.~\cite{tong2021fast} and relies on a dedicated arithmetic operator 
$B$. Following Lemma~10.10 of~\cite{chakraborty2024implementing}, the unitary 
$B$ implements a $(t+1)$-digit approximation of 
$g: x \mapsto \frac{1}{\pi}\arcsin(k/x)$ with depth $O(\operatorname{polylog}(N) + t)$, 
acting as
$$
    B O_D\ket{0}_{t+2} \ket{i} = \ket{i} \ket{\operatorname{sgn}(D_{ii})} 
    \ket{g\!\left(\frac{|\tilde{D}_{ii}|}{M}\right)}_{t+1}.
$$
Applying the rotation $U_\theta$ and uncomputing $B^\dagger$, $O_D^\dagger$ 
embeds the reciprocal $1/D_{ii}$ into the amplitude of the ancilla qubit, 
thereby producing the desired block encoding of $D^{-1}$ with precision 
$\varepsilon/(\pi M)$, as established in the complete derivation of 
Ref.~\cite{antonioli2026quantumblockencodingsemiseparable}.
\end{proof}

The complete quantum circuit implementing this protocol, including the 
arithmetic operator $B$ and the controlled rotation $U_\theta$, is depicted 
in Figure~\ref{fig:block_encoding_circuit}.

\begin{figure}[htpb]
    \centering
    \begin{align}   
\Qcircuit @C=1em @R=.7em {
\lstick{\ket{0}} & \qw&\qw  &\qw  &\qw& \qw &\qw&\qw&\qw &\multigate{1}{U_\theta} &\qw & \qw   & \qw &\qw &\qw\\
\lstick{\ket{0}}& {/}^{t'} \qw &\qw&\qw&\multigate{1}{O_D}     & \qw   &\multigate{1}{B}   &\qw  &\qw   &\ghost{U_\theta}   & \qw& \multigate{1}{B^\dagger} &\qw& \multigate{1}{O_D^\dagger}    &\qw\\
\lstick{\ket{i}}   &{/}^{n} \qw &\qw&\qw& \ghost{O_D} & \qw  &   \ghost{B}   &\qw &\qw&\qw &\qw&\ghost{B^\dagger} &\qw& \ghost{O_D^\dagger} & \qw
}
    \end{align}
    \caption{Quantum circuit for the block encoding of the inverse diagonal operator $D^{-1}$. The circuit uses an oracle $O_D$ to load the diagonal entries, an arithmetic unitary $B$ to compute the required rotation angles, and the controlled rotation unitary $U_\theta$ to encode the reciprocals into amplitudes.}
    \label{fig:block_encoding_circuit}
\end{figure}
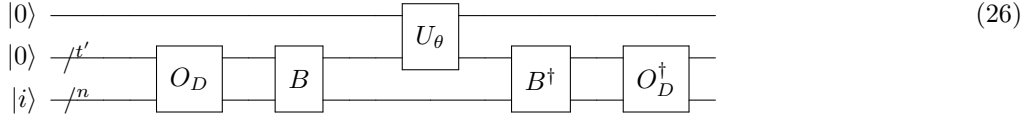

\section{Numerical experiments}
\label{sec:numerical_experiments}


In this section, we describe the implementation of our framework and provide an empirical demonstration of its performance.
We denote $u_{true}$ as the \emph{reference solution}, which we consider to be the ground truth. This solution is obtained using a standard classical spectral method via the Fast Fourier Transform (FFT). We refer to the Kronecker-based spectral method introduced in Section~\ref{sec:spectral_method} as the \emph{classical method}, while the proposed subroutine in Section~\ref{sec:quantum_block_encoding} is referred to as the \emph{quantum method}. Both methods exploit the Kronecker formulation, and we compare their approximate solutions against the reference solution. 

The implementation and simulation of the quantum methods were performed using the Qiskit framework \cite{qiskit2024}. Our benchmarks focus on the following components:

\begin{itemize}
    \item \textbf{Quantum Circuit Implementation:} We provide illustrations of the quantum circuits and analyse the gate complexity for the spectral filter application.
    
    \item \textbf{Elliptic Equation Benchmarks:} We evaluate the accuracy of the classical numerical solution ($u_{num}$) and the quantum solution ($u_{quant}$) by computing their relative $L^2$-errors against the FFT-based reference solution:
    \begin{equation}
     \varepsilon_{num} = \frac{\| u_{num} - u_{true} \|_2}{\| u_{true} \|_2} \quad \text{and} \quad  
     \varepsilon_{quant} = \frac{\| u_{quant} - u_{true} \|_2}{\| u_{true} \|_2}.
     \label{eq:reL_err_def}
     \end{equation}
    \item \textbf{Anisotropic Diffusion Equations:} We compute the relative error~\eqref{eq:reL_err_def} at the steady state, where $u_{true}$ remains the same as in the corresponding elliptic case. Additionally, we track the evolution of the energy functional $E(u)$ defined in Eq.~\eqref{eq:energy_functional} to ensure the dissipative properties of the flow are correctly captured by both the classical and quantum solvers.
\end{itemize}
The code used for the numerical experiments presented in this work is publicly available at \url{https://github.com/Giacomo-Antonioli/QPDE}.

\subsection{Quantum Circuit Illustration}

We describe here the circuit architecture underlying our quantum solver and the methodology used to simulate it. Before presenting the circuits, we address a practical point concerning the simulation of the block-encoding unitary $U_{\hat{\mathcal{G}}}$. A direct implementation and simulation of Theorem~\ref{theo:invBE} in Qiskit is intractable for the grid sizes considered here. The following proposition justifies replacing it by an ideal unitary that encodes an equivalent matrix, without affecting any observable output of the algorithm.
\begin{prpstn}[Logical equivalence after uncomputation]
\label{prop:sim_assumption}
Let $D \in \mathbb{R}^{N \times N}$ be an invertible diagonal matrix and 
$U_{\mathrm{full}}$ the $(n+1+t')$-qubit unitary as defined in Theorem~\ref{theo:invBE}. 
After uncomputing the $t'$ precision qubits, the residual action on the 
logical-plus-ancilla subspace is a unitary $U_{\mathrm{circuit}}$ that block-encodes 
$A = D^{-1}$ with normalisation factor $\alpha = c/m$, up to the arithmetic error 
$\varepsilon$ established in Theorem~\ref{theo:invBE}. That is,
\[
\Bigl\| A - \alpha \bigl(\langle 0| \otimes I\bigr) U_{\mathrm{circuit}} 
\bigl(|0\rangle \otimes I\bigr) \Bigr\|_2 \leq \varepsilon.
\]

The precise entries of the off-diagonal ``garbage'' blocks in $U_{\mathrm{circuit}}$ 
depend on the details of the arithmetic synthesis and are not uniquely determined. 
However, any unitary $U_A$ that shares the same top-left block (up to $\varepsilon$) 
is functionally equivalent for our purposes, since only the post-selected action on 
the $\ket{0}$ ancilla branch is ever accessed. We may therefore replace 
$U_{\mathrm{circuit}}$ by the equivalent block-encoding unitary
$$
U_A = \begin{pmatrix} A/\alpha & (I - A^\dagger A/\alpha^2)^{1/2} \\ 
(I - A^\dagger A/\alpha^2)^{1/2} & -A^\dagger/\alpha \end{pmatrix},
$$
which guarantees unitarity and reproduces the same encoded matrix $A/\alpha$, without 
altering any observable outcome of the algorithm.
\end{prpstn}

Rather than simulating Theorem~\ref{theo:invBE} gate by gate, we compute $U_A$ classically and load it directly into a noiseless Qiskit simulator~\cite{qiskit2024} via a Unitary Gate. The full circuit, comprising the QFTs, $U_A$, and their composition, is then simulated exactly at machine precision ($\varepsilon \approx$ machine epsilon). This quantum simulation validates the logical composition of the QFT and the block-encoded filter against classical spectral methods. The column labelled ``Quantum Simulation'' in tables below therefore reports errors in the limit $\varepsilon \to 0$ of Theorem~\ref{theo:invBE}. Our framework is intended as a subroutine for fault-tolerant quantum architectures with sufficient logical qubits, where the ancilla overhead and circuit depth of
Theorem~\ref{theo:invBE} are well within reach and $\varepsilon$ can be made
small enough that the total error is dominated by the discretisation error.
No resource estimation or hardware noise model is implied by the present
numerical results.

\textbf{Implementation.}
The solver operates on $dn$ register qubits (where $d$ is the spatial dimension and 
$N = 2^n$ points per dimension) and a single ancilla qubit. We assume the source 
term is provided as a quantum state $\ket{f}$ from a preceding quantum routine. 
The circuit therefore begins from the state
$$
\ket{\psi_1} = \ket{0} \otimes \sum_{x_1, \dots, x_d \in \mathbb{Z}_{2^n}} 
f(x_1, \dots, x_d) \ket{x_1} \dots \ket{x_d}.
$$

Since the filter that we aim to block-encode in the circuit is expressed directly in the frequency domain, in order to apply it to the lattice we must transform it into the same domain. In order to do so, we apply the 
$\mathrm{QFT}^{\otimes d}_n = \bigotimes_{i=1}^d \mathrm{QFT}_n$, the $dn-$dimensional QFT, whose circuit is reported in Figure~\ref{fig:CONV_QFT}, which yields the new state:

\begin{figure}[h]
\centering
\scalebox{0.9}{
$
\Qcircuit @C=0.4em @R=.7em {
\lstick{q_{0}} & \multigate{7}{\mathrm{QFT^{\otimes 2}_4}} & \qw & & & & & & 
& & & & & \lstick{q_{0}} & \multigate{3}{\mathrm{QFT}_4} & \qw & & & & & & & 
& \lstick{q_{0}} & \gate{H} & \gate{P} & \gate{P} & \gate{P} & \qw & \qw & 
\qw & \qw & \qw & \qw & \qswap\qwx[3] & \qw & \qw \\
\lstick{q_{1}} & \ghost{\mathrm{QFT^{\otimes 2}}_4} & \qw & & & & & & & & & 
& & \lstick{q_{1}} & \ghost{\mathrm{QFT}_4} & \qw & & & & & & & & 
\lstick{q_{1}} & \qw & \ctrl{-1} & \qw & \qw & \gate{H} & \gate{P} & \gate{P} 
& \qw & \qw & \qw & \qw & \qswap\qwx[1] & \qw \\
\lstick{q_{2}} & \ghost{\mathrm{QFT^{\otimes 2}}_4} & \qw & & & & & & & & & 
& & \lstick{q_{2}} & \ghost{\mathrm{QFT}_4} & \qw & & & & & & & & 
\lstick{q_{2}} & \qw & \qw & \ctrl{-2} & \qw & \qw & \ctrl{-1} & \qw & 
\gate{H} & \gate{P} & \qw & \qw & \qswap & \qw \\
\lstick{q_{3}} & \ghost{\mathrm{QFT^{\otimes 2}}_4} & \qw & & & & & & & & & 
& & \lstick{q_{3}} & \ghost{\mathrm{QFT}_4} & \qw & & & & & & & & 
\lstick{q_{3}} & \qw & \qw & \qw & \ctrl{-3} & \qw & \qw & \ctrl{-2} & \qw 
& \ctrl{-1} & \gate{H} & \qswap & \qw & \qw \\
\lstick{q_{4}} & \ghost{\mathrm{QFT^{\otimes 2}}_4} & \qw & & & & = & & & & & 
& & \lstick{q_{4}} & \multigate{3}{\mathrm{QFT}_4} & \qw & &   =&& & & & & 
\lstick{q_{4}} & \gate{H} & \gate{P} & \gate{P} & \gate{P} & \qw & \qw & 
\qw & \qw & \qw & \qw & \qswap\qwx[3] & \qw & \qw \\
\lstick{q_{5}} & \ghost{\mathrm{QFT^{\otimes 2}}_4} & \qw & & & & & & & & & 
& & \lstick{q_{5}} & \ghost{\mathrm{QFT}_4} & \qw & & & & & & & & 
\lstick{q_{5}} & \qw & \ctrl{-1} & \qw & \qw & \gate{H} & \gate{P} & \gate{P} 
& \qw & \qw & \qw & \qw & \qswap\qwx[1] & \qw \\
\lstick{q_{6}} & \ghost{\mathrm{QFT^{\otimes 2}}_4} & \qw & & & & & & & & & 
& & \lstick{q_{6}} & \ghost{\mathrm{QFT}_4} & \qw & & & & & & & & 
\lstick{q_{6}} & \qw & \qw & \ctrl{-2} & \qw & \qw & \ctrl{-1} & \qw & 
\gate{H} & \gate{P} & \qw & \qw & \qswap & \qw \\
\lstick{q_{7}} & \ghost{\mathrm{QFT^{\otimes 2}}_4} & \qw & & & & & & & & & 
& & \lstick{q_{7}} & \ghost{\mathrm{QFT}_4} & \qw & & & & & & & & 
\lstick{q_{7}} & \qw & \qw & \qw & \ctrl{-3} & \qw & \qw & \ctrl{-2} & \qw 
& \ctrl{-1} & \gate{H} & \qswap & \qw & \qw \\
}
$
}
\caption{$\mathrm{QFT}^{\otimes d}_n$ for $d=2$ and $n=4$}
\label{fig:CONV_QFT}
\end{figure}
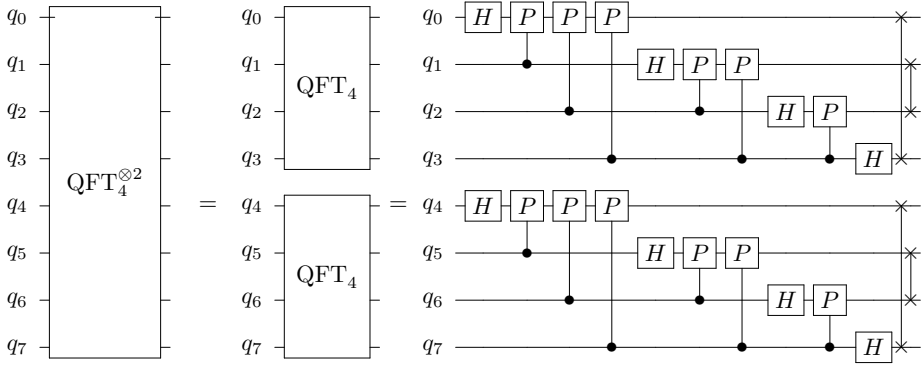

\begin{equation}
\ket{\psi_2} = \ket{0} \otimes \frac{1}{\sqrt{2^{dn}}} \sum_{k_1, \dots, k_d 
\in \mathbb{Z}_{2^n}} \hat{f}(k_1, \dots, k_d) \ket{k_1} \dots \ket{k_d}
\end{equation}


Next, we apply the subcircuit that generates the operator allowing us to 
block encode, starting from QRAM, the inverse of a diagonal matrix. We 
use this to encode the elements of the spectral filter $\hat{\mathcal{G}}$ 
defined in Section~\ref{subsec:elliptic_pdes}. 

To illustrate our approach, consider a 2D elliptic equation ($d=2$) with a 
symmetric positive-definite coefficient matrix $A$. Let $D_N(k_i)$ denote 
the $k_i$-th diagonal element of the 1D derivative matrix $D_N$. For  
each frequency mode $(k_1, k_2) \in \mathbb{Z}^2$, the diagonal elements
of the filter evaluated using the tensor structure from \eqref{eq:bf_D} are given by
\begin{equation}
\hat{\mathcal{G}}(k_1, k_2) = \frac{1}{A_{11} D_N(k_1)^2 + 2A_{12} D_N(k_1) 
D_N(k_2) + A_{22} D_N(k_2)^2}
\label{eq:elliptic_filter_mode}
\end{equation}

Regarding the general $d$-dimensional case, since the filter is defined as a diagonal matrix of size $2^{dn} \times 2^{dn}$, and since the lattice is loaded in the amplitudes in its vectorized form, the resulting state is entangled with the ancilla qubit as follows:
\begin{equation}
\ket{\psi_3} = \ket{0} \otimes \frac{1}{\sqrt{2^{dn}}} \sum_{k_1, \dots, k_d 
\in \mathbb{Z}_{2^n}} \hat{\mathcal{G}}(k_1, \dots, k_d) \hat{f}(k_1, \dots, 
k_d) \ket{k_1} \dots \ket{k_d} + \ket{1} \otimes \ket{*}
\end{equation}

In other words, the result can be rewritten as:
\begin{equation}
\ket{\psi_4} = \ket{0} \otimes \sum_{x_1, \dots, x_d \in \mathbb{Z}_{2^n}} 
u(x_1, \dots, x_d) \ket{x_1} \dots \ket{x_d} + \ket{1} \otimes \ket{*}
\end{equation}
where $u(x_1, \dots, x_d)$ is the target solution. The complete circuit is shown in Figure~\ref{fig:poisson_circuit}.

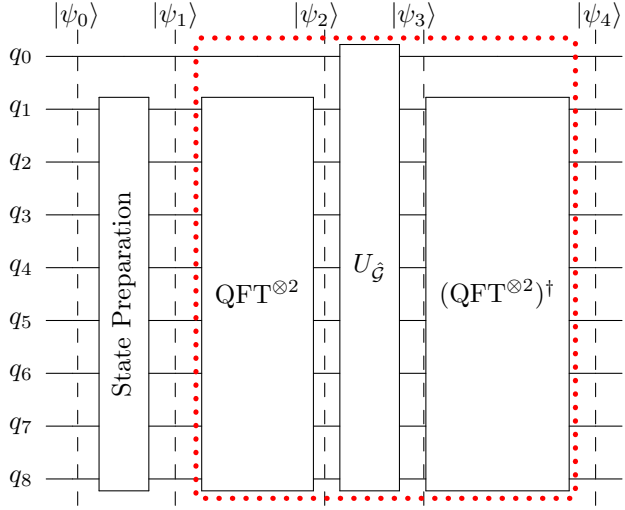
\begin{figure}
    \centering
    \newsavebox{\mycircuit}
    \savebox{\mycircuit}{%
\begin{minipage}{1.5\textwidth}
\centering
\Qcircuit @C=1em @R=1.1em {
&\;\ket{\psi_0}&&\ket{\psi_1}&\;\;\;\;\;\;\;\;\;\;\;\;\;\;\;\;\ket{\psi_2}&\;\;\;\;\;\;\;\;\;\;\;\;\ket{\psi_3}&&\;\ket{\psi_4}\\
\lstick{q_{0}} & \barrier[-1.75em]{8}\qw & \qw \barrier[0em]{8} & \qw & \barrier[-1.7em]{8}\qw &  \multigate{8}{U_{\hat{\mathcal{G}}}} & \barrier[-6.5em]{8}\qw \barrier[0em]{8} & \qw & \qw \\
\lstick{q_{1}} & \qw & \multigate{7}{\rotatebox{90}{State Preparation}} & \qw & \multigate{7}{\mathrm{QFT}^{\otimes 2}} & \ghost{U_{\hat{\mathcal{G}}}} & \multigate{7}{(\mathrm{QFT}^{\otimes 2})^{\dagger}} & \qw & \qw \\
\lstick{q_{2}} & \qw & \ghost{\rotatebox{90}{State Preparation}} & \qw & \ghost{\mathrm{QFT}^{\otimes 2}} & \ghost{U_{\hat{\mathcal{G}}}} & \ghost{(\mathrm{QFT}^{\otimes 2})^{\dagger}} & \qw & \qw \\
\lstick{q_{3}} & \qw & \ghost{\rotatebox{90}{State Preparation}} & \qw & \ghost{\mathrm{QFT}^{\otimes 2}} & \ghost{U_{\hat{\mathcal{G}}}} & \ghost{(\mathrm{QFT}^{\otimes 2})^{\dagger}} & \qw & \qw \\
\lstick{q_{4}} & \qw & \ghost{\rotatebox{90}{State Preparation}} & \qw & \ghost{\mathrm{QFT}^{\otimes 2}} & \ghost{U_{\hat{\mathcal{G}}}} & \ghost{(\mathrm{QFT}^{\otimes 2})^{\dagger}} & \qw & \qw \\
\lstick{q_{5}} & \qw & \ghost{\rotatebox{90}{State Preparation}} & \qw & \ghost{\mathrm{QFT}^{\otimes 2}} & \ghost{U_{\hat{\mathcal{G}}}} & \ghost{(\mathrm{QFT}^{\otimes 2})^{\dagger}} & \qw & \qw \\
\lstick{q_{6}} & \qw & \ghost{\rotatebox{90}{State Preparation}} & \qw & \ghost{\mathrm{QFT}^{\otimes 2}} & \ghost{U_{\hat{\mathcal{G}}}} & \ghost{(\mathrm{QFT}^{\otimes2})^{\dagger}} & \qw & \qw \\
\lstick{q_{7}} & \qw & \ghost{\rotatebox{90}{State Preparation}} & \qw & \ghost{\mathrm{QFT}^{\otimes 2}} & \ghost{U_{\hat{\mathcal{G}}}} & \ghost{(\mathrm{QFT}^{\otimes 2})^{\dagger}} & \qw & \qw \\
\lstick{q_{8}} & \qw & \ghost{\rotatebox{90}{State Preparation}} & \qw & \ghost{\mathrm{QFT}^{\otimes 2}} & \ghost{U_{\hat{\mathcal{G}}}} & \ghost{(\mathrm{QFT}^{\otimes 2})^{\dagger}} & \qw & \qw \\ 
}
\end{minipage}%
    }

   \begin{tikzpicture}
    \node[inner sep=0pt] (circuit) {\usebox{\mycircuit}};
    \draw[red, line width=2pt, line cap=round, dash pattern=on 0pt off 5pt, rounded corners=2pt]
        ([xshift=0.261\wd\mycircuit, yshift=0.03\ht\mycircuit] circuit.south west)
        rectangle
        ([xshift=0.92\wd\mycircuit, yshift=1.86\ht\mycircuit] circuit.south west);
\end{tikzpicture}
\caption{Quantum circuit for solving the 2D Poisson equation on a $16 \times 16$ grid using 9 qubits. The circuit consists of a the $\mathrm{QFT}^{\otimes 2}$, followed by the block encoding of the diagonal filter corresponding to the inverse Laplacian, and concludes with an inverse  $\mathrm{QFT}^{\otimes 2}$ to return to the spatial domain.}
    \label{fig:poisson_circuit}
\end{figure}

\noindent Below we provide a proposition summarizing the complexity of the quantum subroutine, and to make explicit the contribution of the QFT and the QBE components, in the total complexity. The bound assumes only that the source state $\ket{f}$ is already available---as is typical when the solver is invoked as a building block within a larger quantum workflow.

\begin{prpstn}[Quantum Spectral Subroutine]
\label{prop:elliptic_subroutine}
Let $\mathcal{L}$ be a second-order elliptic differential operator with constant coefficients on a $d$-dimensional grid with $N = 2^n$ points per dimension, and let $\kappa$ and $M$ denote its condition number and maximum spectral value. Given a $dn$-qubit normalized state $\ket{f}$ encoding the source term, the solution state $\ket{u}$ can be prepared to precision $\varepsilon > 0$ with circuit depth
$$
\mathcal{O}\Bigl( n^2 + \operatorname{polylog}(dN) + \log\!\bigl(\tfrac{\kappa M}{\varepsilon}\bigr) \Bigr).
$$
where $\mathcal{O}(n^2)$ comes from the parallel application of the QFTs and $\mathcal{O}\Bigl(\operatorname{polylog}(dN) +\log\!\bigl(\tfrac{\kappa M}{\varepsilon}\bigr) \Bigr)$ comes from the block encoding of the operator.
\end{prpstn}

\begin{rmrk}[Input and output assumptions]
The cost of preparing $\ket{f}$ from classical data is not included; in the intended use case, $\ket{f}$ is produced by a preceding quantum routine or loaded via an efficient state-preparation primitive exploiting QRAM. Conversely, extracting the full classical solution $u$ from $\ket{u}$ via repeated measurements would require $\mathcal{O}(N^d/\delta^2)$ samples to achieve precision $\delta$, erasing any exponential speedup. The framework therefore delivers its full advantage when the output state $\ket{u}$ is consumed directly by a subsequent quantum operation.
\end{rmrk}

\subsection{Elliptic PDEs with constant coefficients}
\label{subsec:elliptic_benchmark}

We compare the classical method and the proposed quantum subroutine method against the classical FFT method, used as reference solution, for elliptic PDEs with constant coefficient in 2D and 3D. For each configuration we fix the dimension of the grid $N$ and we test for different coefficient matrices $A$. In particular we focus on how their condition number affects the solution, by testing progressively more ill-conditioned matrices. 
The simulation results for elliptic PDEs with constant coefficients~\eqref{eq:elliptic_PDE_with_constant_coeffs} are reported in Tables~\ref{tab:elliptic2d}--\ref{tab:elliptic3d}, visualization of the approximate solutions and the error distributions with respect to the reference solution are given in Figures~\ref{fig:elliptic2d}--~\ref{fig:elliptic3d}. 

\begin{table}[htbp]

    \caption{Relative errors for  classical and quantum methods for the 2D elliptic  equation~\eqref{eq:elliptic_PDE_with_constant_coeffs} with $N=64$ and the source term $f(x, y) := \cos(2\pi x) \sin(-4 \pi y)$. Condition numbers refer to the matrix $A$.
    }
    \label{tab:elliptic2d}
    \centering
    \begin{tabular}{lccc}
        \hline
        $A$ & Cond. Num. & Numerical Error & Quantum Error \\
        \hline
        $I_2$ & 1 & $1.77\times10^{-15}$ & $2.32\times10^{-15}$ \\
        $\begin{pmatrix} 3 & 1 \\ 1 & 2 \end{pmatrix}$ & 2.62 & $2.25\times10^{-15}$ & $2.59\times10^{-15}$ \\
        diag($10^1, 1$) & $10^1$ & $2.87\times10^{-15}$ & $3.21\times10^{-15}$ \\
        diag($10^2, 1$) & $10^2$ & $1.83\times10^{-14}$ & $2.96\times10^{-14}$ \\
        diag($10^2, 10^{-1}$) & $10^3$ & $1.96\times10^{-14}$ & $6.29\times10^{-14}$ \\
        diag($10^5, 1$) & $10^5$ & $1.75\times10^{-11}$ & $2.20\times10^{-11}$ \\
        \hline
    \end{tabular}
\end{table}

\begin{table}[htbp]

    \caption{Relative errors for  classical and quantum methods for the 3D elliptic  equation~\eqref{eq:elliptic_PDE_with_constant_coeffs} with $N=16$ and the source term $f(x, y, z):=\cos(2\pi x) \sin(-4\pi y) \cos (2\pi z)$.
    Condition numbers refer to the matrix $A$.
    }
    \label{tab:elliptic3d}
    \centering
    \begin{tabular}{lccc}
        $A$  & Cond. Num. & Numerical Error  & Quantum Error \\
        \hline
        $I_3$ 
            & 1 & $2.24\times 10^{-14}$ & $ 2.93\times 10^{-15}$ \\
        $\begin{pmatrix} 3 & 1 & 0.5 \\ 1 & 3 & 1 \\ 0.5 & 1 & 3 \end{pmatrix}$ 
            & 2.58& $1.2232\times 10^{-15}$ & $ 3.12\times 10^{-15}$ \\
        diag($10$, $1$, $1$) 
            & 10 & $2.95\times 10^{-15}$ & $ 6.27\times 10^{-15}$ \\
        diag($1$, $10^2$, $1$) 
            & $10^{2}$ & $8.38\times 10^{-15}$ & $3.47\times 10^{-14}$ \\
        diag($1$, $10^2$, $10^{-1}$) 
            & $10^{3}$ & $2.65\times 10^{-14}$ & $3.72\times 10^{-14}$ \\
        diag($1$, $1$, $10^5$) 
            & $10^5$ & $ 2.10\times 10^{-11}$ & $ 2.80\times 10^{-11}$ 
            \\ \hline
    \end{tabular}
\end{table}

\noindent The results reported in Tables~\ref{tab:elliptic2d} and \ref{tab:elliptic3d} show a clear relationship between the  condition number of the matrix $A$ and the relative errors observed in both the classical numerical method and our quantum method simulated on Qiskit.  For the first three matrices, where the condition numbers are less than $10$, both methods show relatively low errors compared to the reference solution. As the condition number increases further, the accuracy of both approaches degrades, following a similar trend. It is also important to note that the quantum simulated error is always worse than the classical numerical one.  This can be due to the fact that a quantum simulator like Qiskit must compute long sequences of matrix-vector or matrix-matrix multiplications and tensor products to represent consecutive quantum gate operations, as a result, standard floating-point round-off errors naturally accumulate at each step, artificially raising the noise floor.

Figures~\ref{fig:elliptic2d} and \ref{fig:elliptic3d} reveal an interesting difference about the error distribution. While the errors from the classical method align with the structure of the reference solution on a frequency basis, the errors from the quantum method are randomly distributed across the domain with no visible correlation to the profile of the reference solution. This observation suggests and confirms that the quantum error is primarily dependent on the simulation environment, due to the accumulation of round-off noise.
\begin{figure}[htbp]
    \centering
    \includegraphics[width=0.75\linewidth]{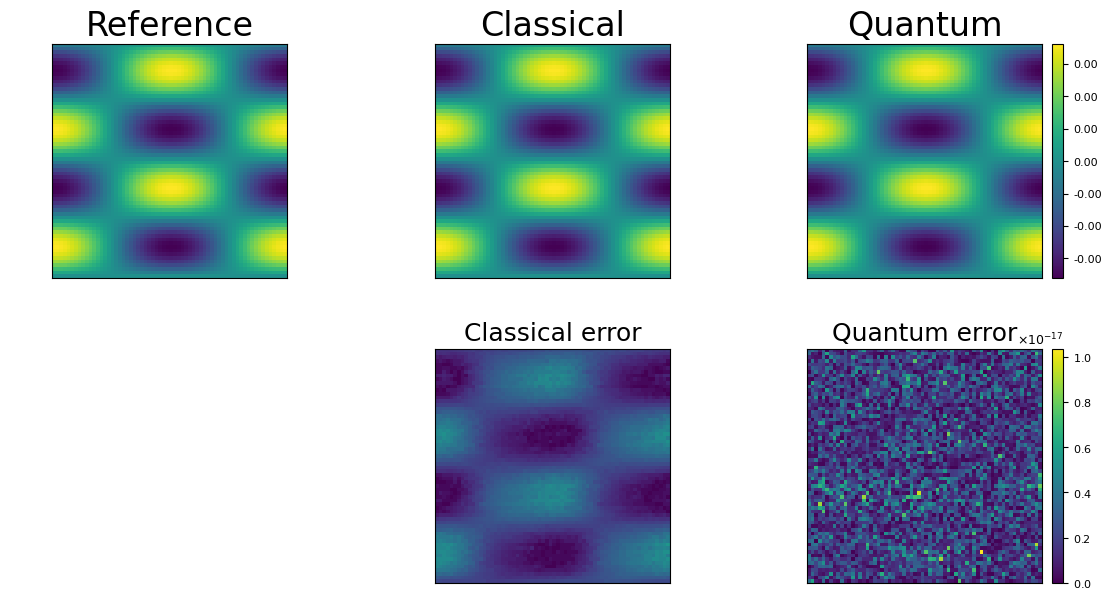}
    \caption{Visualizations of the reference, classical, and quantum solutions, along with the absolute errors of the classical and quantum methods relative to the reference for the 2D elliptic case with $N=64$ and $A=$diag($10$, $1$).}
    \label{fig:elliptic2d}
\end{figure}
\begin{figure}[htbp]
    \centering
    \includegraphics[width=0.75\linewidth]{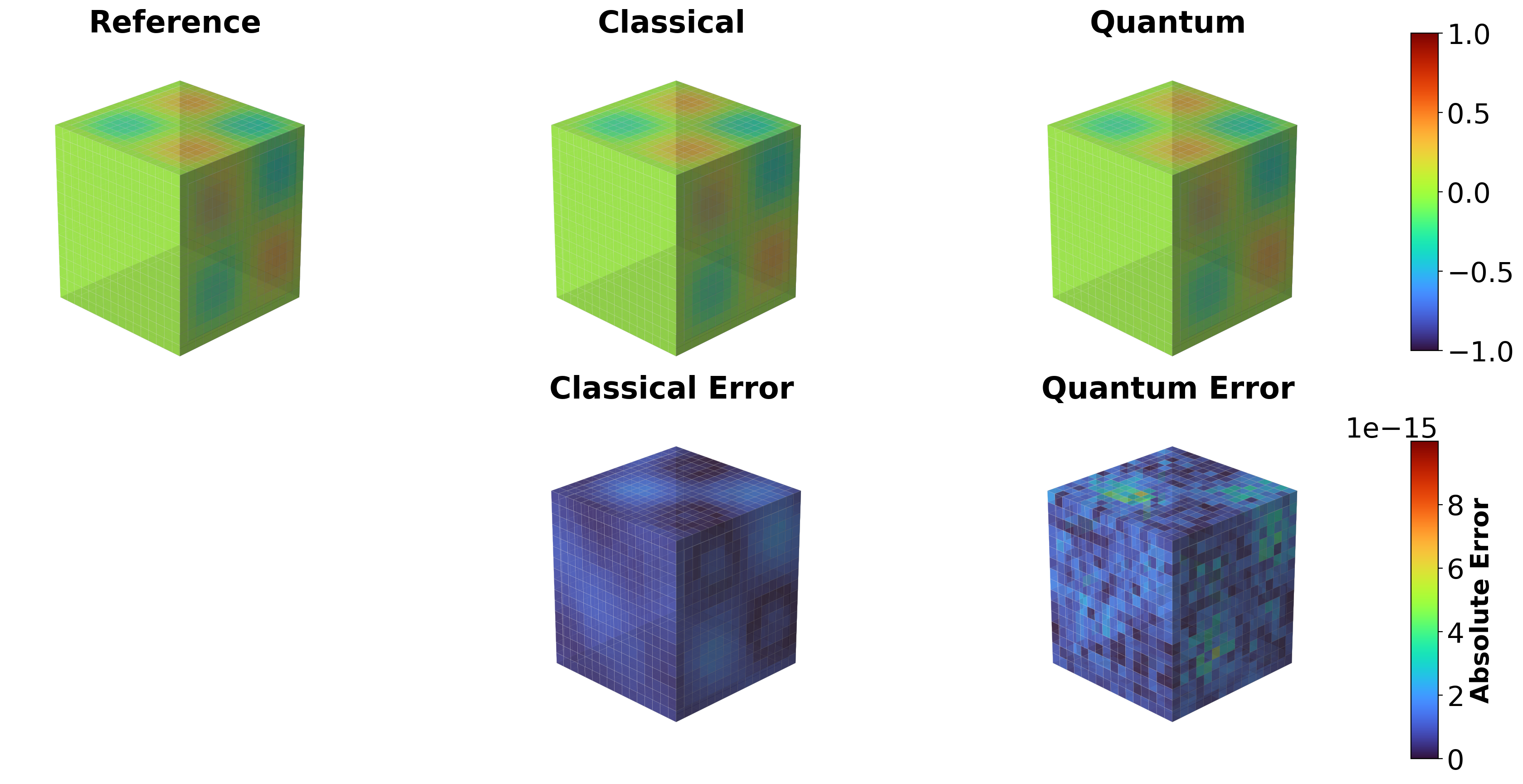}
    \caption{Visualizations of the reference, classical, and quantum solutions, along with the absolute errors of the classical and quantum methods relative to the reference for the 3D elliptic case with $N=16$ and $A=$diag($10$, $1$, $1$).}
    \label{fig:elliptic3d}
\end{figure}

\subsubsection{Helmholtz equation}
The application of the subroutine to the Helmholtz equation is straightforward, since adapting the quantum framework in Section~\ref{subsec:elliptic_benchmark} to other
elliptic PDEs requires no changes to the overall circuit architecture. One only needs to replace the elliptic spectral filter $\hat{\mathcal{G}}$ defined in~\eqref{eq:elliptic_filter_mode} with the Helmholtz filter defined in~\eqref{eq:helmholtz_filter} in Section~\ref{par:Helm}. Considering for example the 2D case, the Helmholtz filter on every frequency mode $(k_1, k_2)$ is given by
\begin{equation}
\hat{\mathcal{G}}_{\mathrm{Helmholtz}}(k_1, k_2) = \frac{1}{D_N(k_1)^2 + D_N(k_2)^2 + \lambda^2}.
\end{equation}
Since the structural flow of the quantum circuit remains identical, the complexity bound established in the proposition~\ref{prop:elliptic_subroutine} applies directly.

\begin{table}[htbp]
\caption{Relative errors for  classical and quantum methods for the 2D Helmholtz equation~\eqref{eq:helmholtz_equation} with $f(x,y):=\cos(2\pi x) \sin(-4 \pi y)$. Condition numbers refer to the Helmholtz filter~\eqref{eq:helmholtz_filter}.}
\centering
\label{tab:helmoltz}
\begin{tabular}{lccc}
    \hline 
    $\lambda$  & Cond. Num. & Numerical sim. & Quantum Sim. \\
    \hline
    $2\pi \cdot 0.5$ & $4.70 \times 10^3$ & $1.35 \times 10^{-15}$ & $1.96 \times 10^{-15}$ \\
    $2\pi \cdot 10^{-1}$ & $1.04\times 10^4$ & $1.97 \times 10^{-15}$ & $2.04 \times 10^{-15}$ \\
    $2\pi \cdot 10^{-2}$ & $1.04\times 10^6$ & $3.29 \times 10^{-14}$ & $1.16 \times 10^{-13}$ \\
    $2\pi \cdot 10^{-3}$ & $1.04\times 10^8$ & $3.30 \times 10^{-12}$ & $1.13 \times 10^{-11}$ \\
    $2\pi \cdot 10^{-4}$ & $1.04\times 10^{10}$ & $3.30 \times 10^{-10}$ & $ 1.10\times 10^{-9}$ \\
    \hline
\end{tabular}
\end{table}

We consider the 2D Helmholtz equation with $n=6$ ($N=64$). As $\lambda \to 0$, the filter on the zero-frequency mode $\hat{\mathcal{G}}_{\text{Helmholtz}}(0,0)$ converges to $0$ and thus the condition number of the spectral filter $\hat{\mathcal{G}}_{\text{Helmholtz}}$ diverges to infinity.
Relative errors for this case are presented in Table~\ref{tab:helmoltz}. Similarly to the Elliptic results, the errors scale as $\lambda$ decreases and the condition number grows. Visualizations of the solutions and error distributions are shown in Figure~\ref{fig:Helm2d}. As observed in the elliptic case, while the classical method produces structured errors that are more pronounced on a frequency basis, the quantum method results in a more randomly distributed error field.

\begin{figure}[htbp]
    \centering
    \includegraphics[width=.49\textwidth]{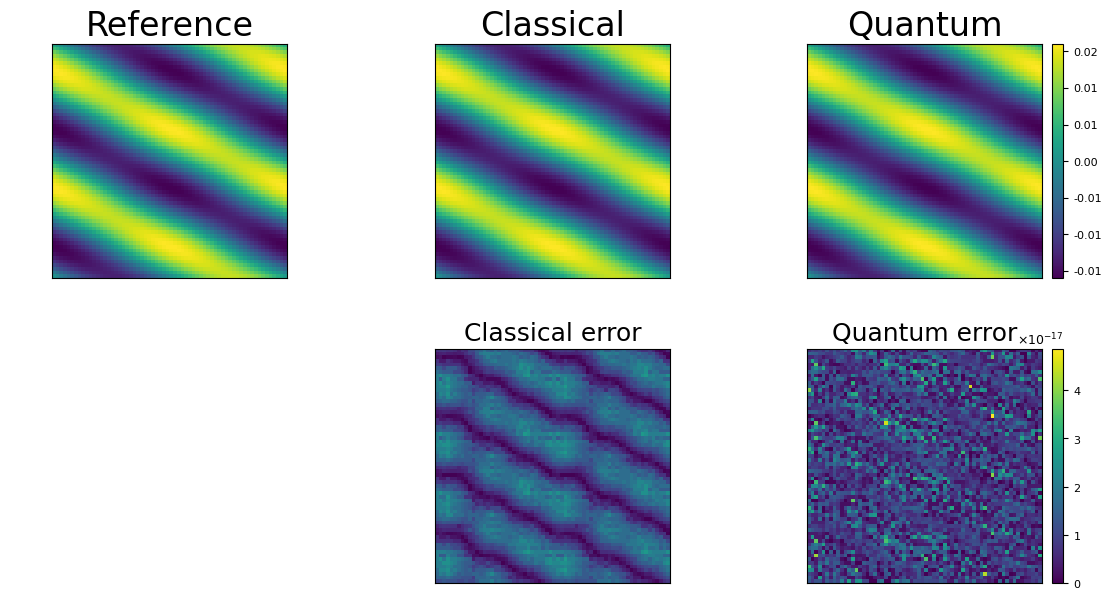}
    \includegraphics[width=.49\textwidth]{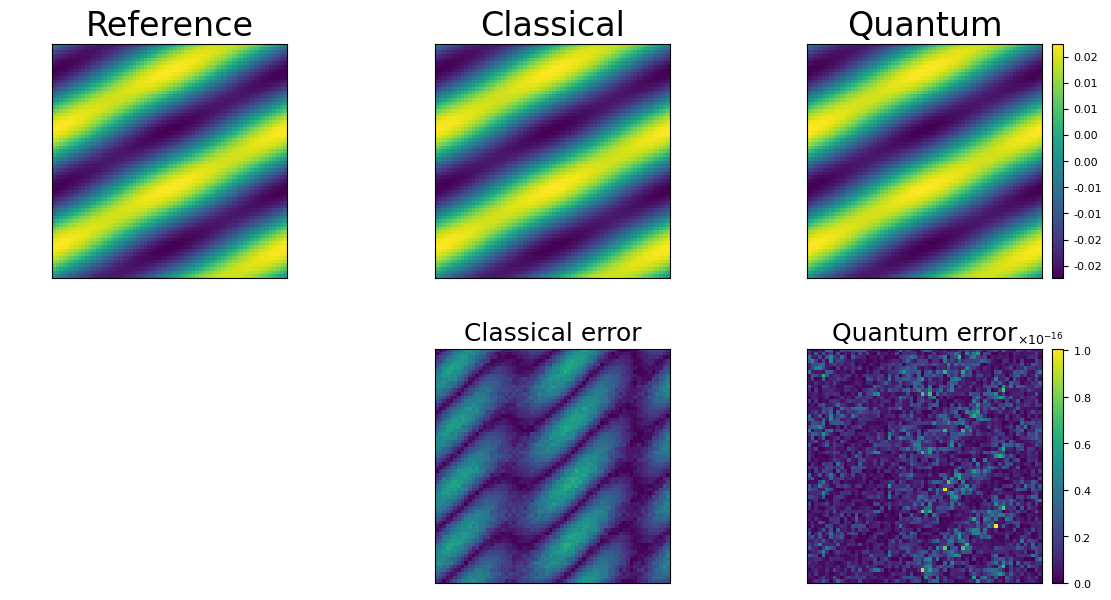}
       \caption{Visualizations of the reference, classical, and quantum solutions, along with the absolute errors of the classical and quantum methods relative to the reference for the 2D Helmholtz equation with $N=64$, $k=2\pi\times0.5$ (Left) and $k=2\pi\times0.1$ (Right).}
    \label{fig:Helm2d}
\end{figure}

\subsection{Anisotropic Diffusion PDEs}

We extend our analysis to the anisotropic diffusion equation~\eqref{eq:diffusion_eq} to evaluate our proposed quantum framework. This application shows that the quantum circuit architecture and its depth bounds given in Proposition~\ref{prop:elliptic_subroutine} remain applicable to time-dependent parabolic PDEs with the corresponding spectral filter.
As introduced in Section~\ref{subsec:parabolic_pdes}, we discretize the time domain into steps of size $\Delta t$ and we employ the implicit time-stepping scheme~\eqref{eq:implicit_scheme}. We then apply the filter iteratively to compute the solution at each time step.

For the 2D anisotropic diffusion equation with a symmetric positive-definite matrix $A$, the diagonal elements of the diffusion filter for each frequency mode $(k_1, k_2)\in\mathbb{Z}^2$ are given by
\begin{equation}
\hat{\mathcal{G}}_{\mathrm{diffusion}}(k_1, k_2) = \frac{1}{1 - \Delta t \left( A_{11} D_N(k_1)^2 + 2A_{12} D_N(k_1) D_N(k_2) + A_{22} D_N(k_2)^2 \right)}
\end{equation}

Unlike for the elliptic equations, solving the anisotropic diffusion equation~\eqref{eq:diffusion_eq} requires an iterative procedure to compute the solution through time. At each iteration step $t$, the solution is given by the implicit scheme defined in Eq.~\eqref{eq:implicit_scheme}. In the spectral domain, this evolution is expressed by Eq.~\eqref{eq:implicit_scheme_spectral}, where the constant spectral filter $\hat{\mathcal{G}}_{\text{diffusion}}$ is applied to the current solution and source term.

From a quantum circuit perspective, this implies that while the filter operator $U_{\hat{\mathcal{G}}}$ remains computationally identical across all iterations, the input state must be updated at each step. Specifically, the quantum state $\ket{u_t}$ representing the solution at the previous time step must be combined with the source term contribution $-\Delta t \ket{f}$ to form the new input state. If this iterative update is performed classically, it requires extracting the full quantum state $\ket{u_t}$ at each time step using Quantum State Tomography (QST)~\cite{cramer2010efficient}. Because QST grows exponentially with the number of qubits, this intermediate measurement bottleneck would eliminate the speedup of our quantum subroutine. 

We now test the time‑dependent diffusion Eq.~\eqref{eq:diffusion_eq} with time step $\Delta t = 10^{-3}$ both for the 2D and 3D case. For each configuration, we classically update the state $\ket{u_t} = \ket{u_n -\Delta t f}$ between iterations, and compute the relative error of the final solution with respect to the \emph{steady-state} solution, which coincides with the solution of the elliptic equation~\eqref{eq:elliptic_PDE_with_constant_coeffs}. Our numerical comparison demonstrates that the proposed method for the implicit time-stepping scheme correctly captures the diffusion evolution. 
However, maintaining a quantum speedup would require keeping the evolution entirely quantum. A possible solution is to employ a LCU scheme: index the $T$ time steps on $\lceil\log_2(T)\rceil$ ancilla qubits and control the application of the spectral filter at each step, thereby performing the state update in a full quantum setting. Integrating our subroutine with such a scheme is left for future research.
    
\begin{table}[htbp]
    \caption{Relative errors for classical and quantum methods
    for the 2D anisotropic diffusion equation~\eqref{eq:diffusion_eq} with ${\Delta t}=10^{-3}$, $n=6$, $N=2^6=64$, source term $f(x,y) := \cos(2\pi x) \sin(-4\pi y)$ and 
    initial condition $u_0(x, y) = \cos(2\pi x) \sin(8\pi y) + 2 \sin(6 \pi y) + 3 \sin(10 \pi x) \cos^2(12 \pi y)$. 
    Condition numbers refer to the diffusion filter~\eqref{eq:diffusion_filter} and errors are averaged at the 300th time step.}
    \centering
    \label{tab:diffusion2d}
\begin{tabular}{lccc}
        \hline
        $A$ & Cond. Num. & Numerical Error & Quantum Error \\
        \hline
        $I_2$ & 81.85 & $4.60 \times 10^{-14}$ & $5.27 \times 10^{-14}$ \\
        $\begin{pmatrix} 3 & 1 \\ 1 & 2 \end{pmatrix}$ & 283.98 & $7.40 \times 10^{-14}$ & $5.67 \times 10^{-14}$ \\
        diag($10^1, 1$) & 445.68 & $1.59 \times 10^{-13}$ & $1.23 \times 10^{-13}$ \\
        diag($10^2, 1$) & $4.08 \times 10^3$ & $1.28 \times 10^{-12}$ & $9.39 \times 10^{-13}$ \\
        diag($10^3, 1$) & $4.05 \times 10^3$ & $1.01 \times 10^{-11}$ & $8.74 \times 10^{-12}$ \\
        diag($10^5, 1$) & $4.04 \times 10^6$ & $1.23 \times 10^{-9}$ & $9.40 \times 10^{-10}$ \\
        \hline
    \end{tabular}
\end{table}

The relative errors for the anisotropic diffusion equation are presented in Table~\ref{tab:diffusion2d} for the 2D case and Table~\ref{tab:diffusion3d} for the 3D case. Similarly to the elliptic case, both errors scale up as the condition number of the diffusion filter~\eqref{eq:diffusion_filter} increases.
The errors produced by the classical and quantum methods are comparable, whereas the classical method does not necessarily perform better because we evaluate the error at a long-time approximation, where errors can accumulate at each iteration of the implicit time-stepping scheme.
\begin{table}[htbp]\label{tab:diff3d}
    \caption{Relative errors for classical and quantum methods for the 3D anisotropic diffusion equation~\eqref{eq:diffusion_eq} with $\Delta t=10^{-3}$, $n=4$, $N=2^4=16$, source term
     $f(x,y, z) := \cos(2\pi x) \sin(-4\pi y)\cos(2\pi z)$ and 
    initial condition $u_0(x, y,z) = \cos(2\pi x) \sin(8\pi y) + 2 \sin(6 \pi y) + 3 \sin(10 \pi x) \cos^2(12 \pi y)$. Condition numbers refer to the diffusion filter~\eqref{eq:diffusion_filter}, and errors are averaged at 400th time step.
    }
    \centering
    \label{tab:diffusion3d}
    \begin{tabular}{lccc}
    \hline
        $A$ & Cond. Num. & Numerical Error & Quantum Error \\
        \hline
        $I_3$ & 31.32 & $3.62 \times 10^{-6}$ & $3.22 \times 10^{-6}$ \\
        $\begin{pmatrix} 3 & 1 & 0.5 \\ 1 & 3 & 1 \\ 0.5 & 1 & 3 \end{pmatrix}$ & 142.49 & $8.16 \times 10^{-6}$ & $7.84 \times 10^{-6}$ \\
        diag($10, 1, 1$) & 122.28 & $7.87 \times 10^{-6}$ & $7.38 \times 10^{-6}$ \\
        diag($1, 10^2, 1$) & $1.03 \times 10^3$ & $2.11 \times 10^{-4}$ & $1.98 \times 10^{-4}$ \\
        diag($1, 1, 10^3$) & $1.02 \times 10^3$ & $5.27 \times 10^{-4}$ & $4.94 \times 10^{-4}$ \\
        diag($1, 1, 10^5$) & $1.01 \times 10^6$ & $3.96 \times 10^{-2}$ & $3.71 \times 10^{-2}$ \\
        \hline
    \end{tabular}
\end{table}

Note that in the 3D case, simulations are performed on a coarser grid with $N=16$, due to computational time and hardware limitation.  In all cases, our simulations show that the quantum subroutine approximates the classical FFT solver. Visualizations and error distributions for both methods are given in Figures~\ref{fig:diffusion2d} and ~\ref{fig:diffusion3d}. In our test cases, the error distributions of the classical method are more pronounced on a frequency basis, although they are less prominent and exhibit greater average errors than in the elliptic case.

\begin{figure}[H]
    \centering
    \includegraphics[width=.75\textwidth]{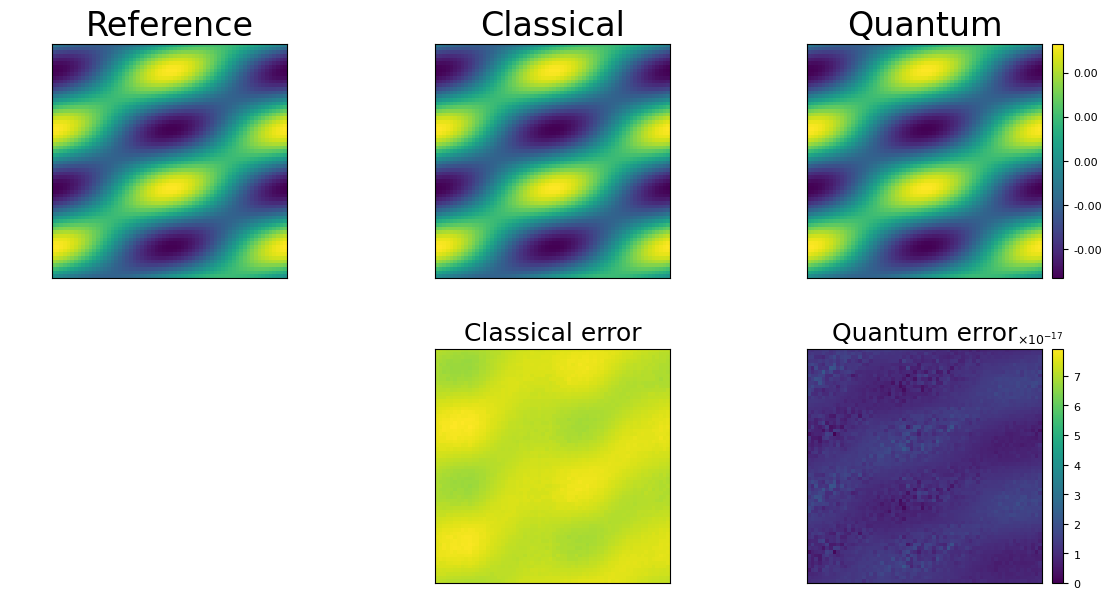}
       \caption{Visualizations of the reference, classical, and quantum solutions, along with the absolute errors of the classical and quantum methods relative to the reference for the 2D diffusion case with $N=64$ and $A=$diag($10$, $1$).}
    \label{fig:diffusion2d}
\end{figure}
\begin{figure}[H]
    \centering
    \includegraphics[width=.75\textwidth]{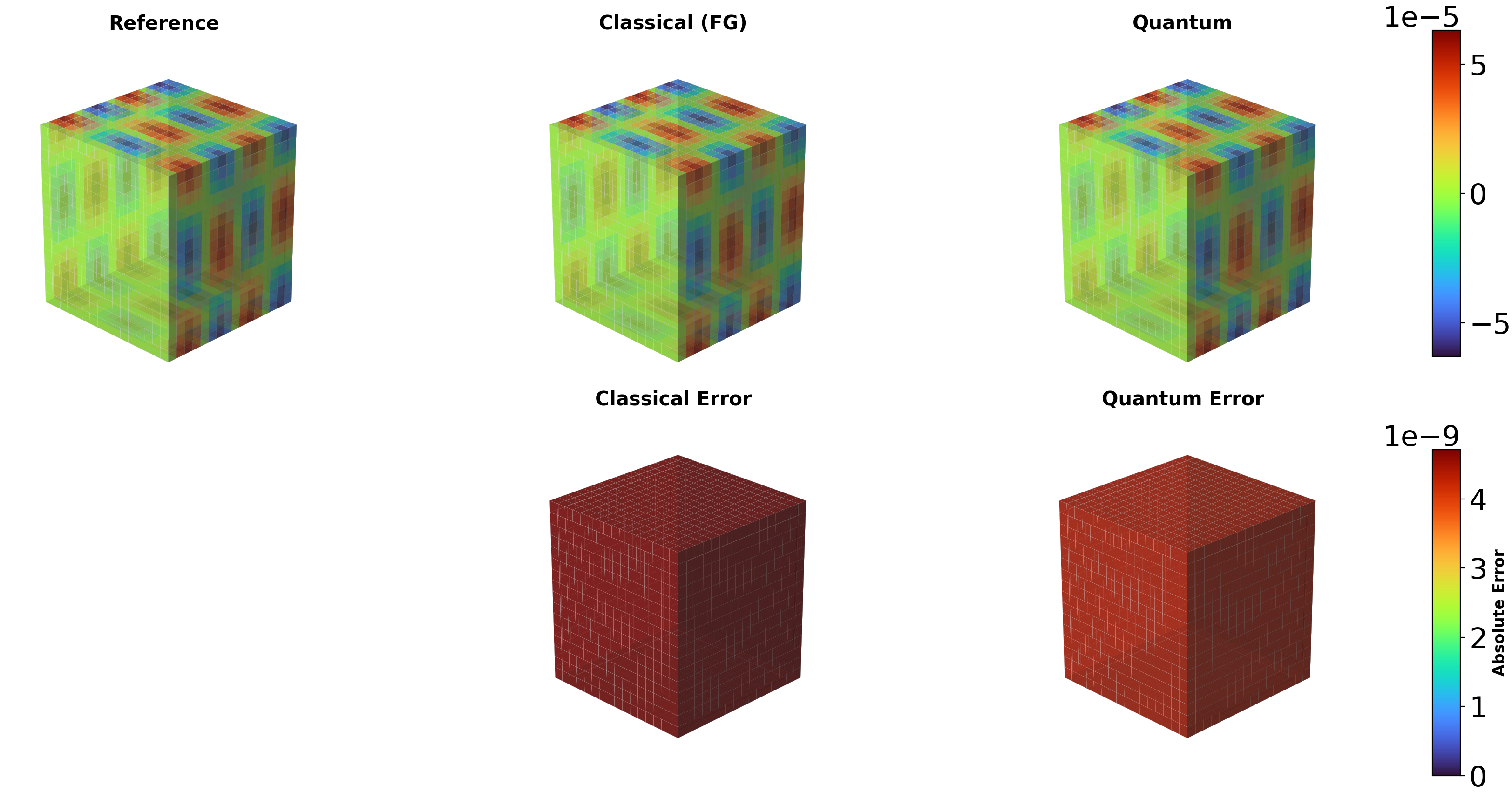}
       \caption{Visualizations of the reference, classical, and quantum solutions, along with the absolute errors of the classical and quantum methods relative to reference for the 3D diffusion case with $N=16$ and $A=$diag($1$,$100$, $1$).}
    \label{fig:diffusion3d}
\end{figure}

\textbf{Energy Dissipation.} To verify the dissipative property of the gradient flow, we track the energy functional $E(u)$ defined by Eq.~\eqref{eq:energy_functional} during its time evolution. Figure~\ref{fig:energy_evolution} shows that $E(u)$ decreases and converges to the energy of the steady-state solution.
Both the classical simulation and the quantum simulation follow the same energy trajectory, confirming that the spectral filter correctly captures the gradient-flow dynamics.

\begin{figure}[htbp]
    \centering
    \includegraphics[width=.49\textwidth]{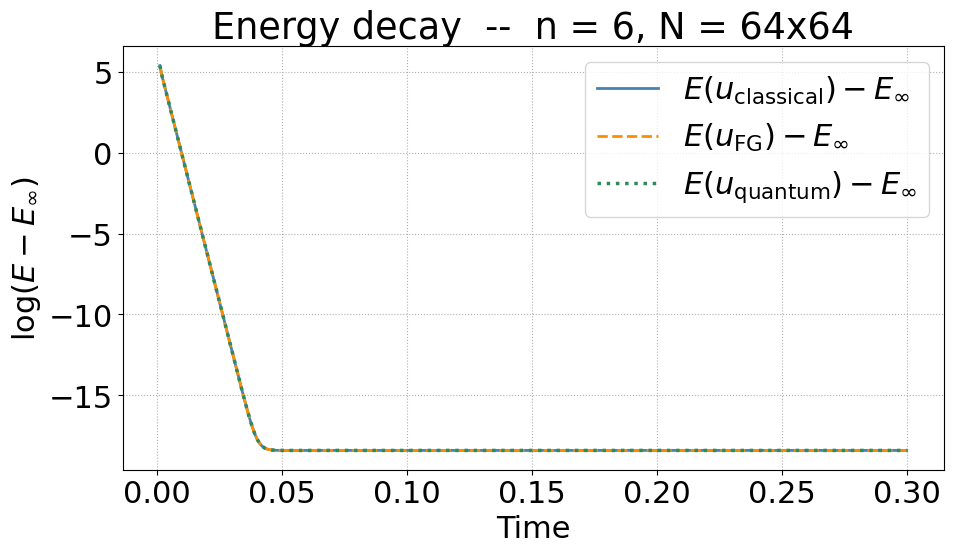}
    \includegraphics[width=.49\textwidth]{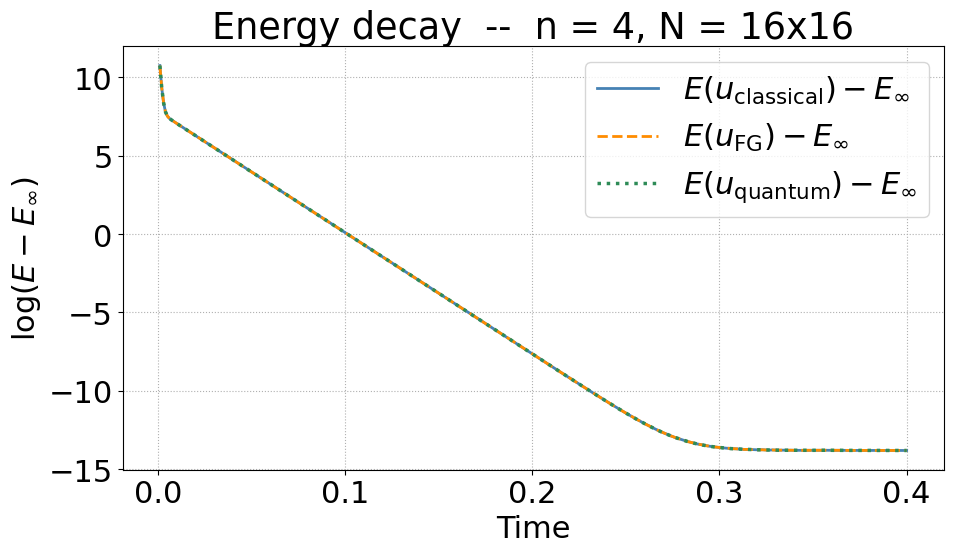}
     \caption{Evolutions and dissipations of the energy~\ref{eq:energy_functional} for the diffusion case with (Left) $d=2$, $N=64$ and $A$=diag($100, 1$);  (Right) $d=3$, $N=16$ and $A$=diag($1, 100, 1$) for the Reference, Classical and Quantum methods. We track $\log (E(u) - E_\infty)$ where $E_\infty$ is the energy of the steady-state solution, since $E$ can become negative during evolution.}
    \label{fig:energy_evolution}
\end{figure}

\clearpage

\section{Conclusion and Perspective}
\label{sec:conclusions}

This work presents a quantum spectral subroutine for solving  elliptic and parabolic PDEs with constant coefficients on periodic domains. We exploit a block-encoding architecture that encodes the inverse spectral filter via QFT and reversible arithmetic. The framework accepts source terms provided as quantum states and returns the solution state with a single ancilla qubit overhead, making it suitable as a subroutine for larger quantum algorithms.
The numerical experiments validate the correctness of our construction. For both elliptic and anisotropic diffusion PDEs, the quantum solver reproduces the classical Kronecker-based spectral method with comparable numerical accuracy. The results confirm that the proposed subroutine implements the intended solution operator. 

Several directions merit further investigation. First, for time-dependent problems, one could consider the extension to a fully quantum iterative solver, via a linear combination of unitaries, indexed by time-step ancillas. Second, one could evaluate the impact of realistic noise on small-scale examples of the proposed circuits. Third, generalizing the construction to  non-constant coefficients would significantly broaden the framework's scope. Beyond Euclidean domains, generalizing QFTs to Quantum Group Fourier Transforms (QGFTs), as introduced by Castelazo et al. in \cite{castelazo2022quantum, motamedi2024group}, could extend the framework to, for instance, Poisson equations on compact Lie groups such as $SO(3)$ with spherical harmonics. Other promising extensions include quantum wavelet-based spectral methods~\cite{wavelet2,wavelet1} and quantum equivariant neural networks~\cite{castelazo2022prospects}. Finally, extending the approach to nonlinear equations represents a promising direction for future study, as it would allow the approach to address an even broader class of problems.

\section*{Acknowledgement}
This work stems from the authors' participation in the Thales research project "Quantum group convolution for PDEs" during the CEMRACS 2025 summer school (\emph{Centre d'été Mathématiques de Recherche Avancée en Calcul Scientifique}). 
The authors are grateful to the organizing committee of CEMRACS and to the CIRM (\emph{Centre International de Rencontres Mathématiques}) for their hospitality.
C.-K. H. also thanks Henrique Ennes and Jui-Ting Lu for fruitful discussion during the preparation of this work at CIRM.

C.-K. H. was supported by the Lorraine Université d'Excellence through the interdisciplinary project TransPINNO during his stay at CIRM, and is currently funded by a France 2030 support managed by the Agence Nationale de la Recherche, under the reference ANR-23-PEIA-0004 (PDE-AI project). 

Experiments presented in this paper were carried out using the Grid'5000 testbed, supported by a scientific interest group hosted by Inria and including CNRS, RENATER and several Universities as well as other organizations (see \url{https://www.grid5000.fr}). 

\bibliographystyle{siam}
\bibliography{references.bib}

@book{boyd2001chebyshev,
title={Chebyshev and Fourier spectral methods},
author={Boyd, John P},
year={2001},
publisher={Courier Corporation}
}

@book{gottlieb1977numerical,
  title={Numerical analysis of spectral methods: theory and applications},
  author={Gottlieb, David and Orszag, Steven A},
  year={1977},
  publisher={SIAM}
}

@book{trefethen2000spectral,
  title={Spectral methods in MATLAB},
  author={Trefethen, Lloyd N},
  year={2000},
  publisher={SIAM}
}

@article{quarteroni2006spectral,
  title={Spectral methods: Fundamentals in single domains},
  author={Quarteroni, A and Canuto, C and Hussaini, MY and Zang, TA},
  journal={Springer Verlag},
  volume={4},
  number={8},
  pages={16},
  year={2006}
}

@book{evans2022partial,
  title={Partial differential equations},
  author={Evans, Lawrence C},
  volume={19},
  year={2022},
  publisher={American mathematical society}
}

@book{steeb2011matrix,
  title={Matrix calculus and Kronecker product: a practical approach to linear and multilinear algebra},
  author={Steeb, Willi-Hans and Hardy, Yorick},
  year={2011},
  publisher={World Scientific}
}

@article{harrow2009quantum,
  title={Quantum algorithm for linear systems of equations},
  author={Harrow, Aram W and Hassidim, Avinatan and Lloyd, Seth},
  journal={Physical review letters},
  volume={103},
  number={15},
  pages={150502},
  year={2009},
  publisher={APS}
}

@book{aubin1998some,
  title={Some nonlinear problems in Riemannian geometry},
  author={Aubin, Thierry},
  year={1998},
  publisher={Springer Science \& Business Media}
}

@book{taylor1996partial,
  title={Partial differential equations. 1, Basic theory},
  author={Taylor, Michael Eugene},
  year={1996},
  publisher={Springer}
}

@article{khoromskij2012tensors,
  title={Tensors-structured numerical methods in scientific computing: Survey on recent advances},
  author={Khoromskij, Boris N},
  journal={Chemometrics and intelligent laboratory systems},
  volume={110},
  number={1},
  pages={1--19},
  year={2012},
  publisher={Elsevier}
}

@book{colton2013integral,
  title={Integral equation methods in scattering theory},
  author={Colton, David and Kress, Rainer},
  year={2013},
  publisher={SIAM}
}

@book{sommerfeld1949partial,
  title={Partial differential equations in physics},
  author={Sommerfeld, Arnold},
  volume={1},
  year={1949},
  publisher={Academic press}
}

@article{etgen2009overview,
  title={An overview of depth imaging in exploration geophysics},
  author={Etgen, John and Gray, Samuel H and Zhang, Yu},
  journal={Geophysics},
  volume={74},
  number={6},
  pages={WCA5--WCA17},
  year={2009},
  publisher={Society of Exploration Geophysicists}
}

@book{hesthaven2007spectral,
  title={Spectral methods for time-dependent problems},
  author={Hesthaven, Jan S and Gottlieb, David I and Gottlieb, Sigal},
  volume={21},
  year={2007},
  publisher={Cambridge University Press Cambridge}
}

@book{strikwerda2004finite,
  title={Finite difference schemes and partial differential equations},
  author={Strikwerda, John C},
  year={2004},
  publisher={SIAM}
}

@article{camps2024explicit,
  title={Explicit quantum circuits for block encodings of certain sparse matrices},
  author={Camps, Daan and Lin, Lin and Van Beeumen, Roel and Yang, Chao},
  journal={SIAM Journal on Matrix Analysis and Applications},
  volume={45},
  number={1},
  pages={801--827},
  year={2024},
  publisher={SIAM}
}

@inproceedings{gilyen2019quantum,
  title={Quantum singular value transformation and beyond: exponential improvements for quantum matrix arithmetics},
  author={Gily{\'e}n, Andr{\'a}s and Su, Yuan and Low, Guang Hao and Wiebe, Nathan},
  booktitle={Proceedings of the 51st annual ACM SIGACT symposium on theory of computing},
  pages={193--204},
  year={2019}
}

@inproceedings{camps2022fable,
  title={Fable: Fast approximate quantum circuits for block-encodings},
  author={Camps, Daan and Van Beeumen, Roel},
  booktitle={2022 IEEE International Conference on Quantum Computing and Engineering (QCE)},
  pages={104--113},
  year={2022},
  organization={IEEE}
}

@article{chakraborty2024implementing,
  title={Implementing any linear combination of unitaries on intermediate-term quantum computers},
  author={Chakraborty, Shantanav},
  journal={Quantum},
  volume={8},
  pages={1496},
  year={2024},
  publisher={Verein zur F{\"o}rderung des Open Access Publizierens in den Quantenwissenschaften}
}

@book{arora2009computational,
  title={Computational complexity: a modern approach},
  author={Arora, Sanjeev and Barak, Boaz},
  year={2009},
  publisher={Cambridge University Press}
}

@article{li2023efficient,
  title={On efficient quantum block encoding of pseudo-differential operators},
  author={Li, Haoya and Ni, Hongkang and Ying, Lexing},
  journal={Quantum},
  volume={7},
  pages={1031},
  year={2023},
  publisher={Verein zur F{\"o}rderung des Open Access Publizierens in den Quantenwissenschaften}
}

@article{lin2022lecture,
  title={Lecture notes on quantum algorithms for scientific computation},
  author={Lin, Lin},
  journal={arXiv preprint arXiv:2201.08309},
  year={2022}
}

@article{qram,
  title = {Architectures for a quantum random access memory},
  author = {Giovannetti, Vittorio and Lloyd, Seth and Maccone, Lorenzo},
  journal = {Phys. Rev. A},
  volume = {78},
  issue = {5},
  pages = {052310},
  numpages = {9},
  year = {2008},
  month = {Nov},
  publisher = {American Physical Society},
  doi = {10.1103/PhysRevA.78.052310},
  url = {https://link.aps.org/doi/10.1103/PhysRevA.78.052310}
}

@article{castelazo2022quantum,
  title={Quantum algorithms for group convolution, cross-correlation, and equivariant transformations},
  author={Castelazo, Grecia and Nguyen, Quynh T and De Palma, Giacomo and Englund, Dirk and Lloyd, Seth and Kiani, Bobak T},
  journal={Physical Review A},
  volume={106},
  number={3},
  pages={032402},
  year={2022},
  publisher={APS}
}

@misc{qiskit2024,
      title={Quantum computing with {Q}iskit},
      author={Javadi-Abhari, Ali and Treinish, Matthew and Krsulich, Kevin and Wood, Christopher J. and Lishman, Jake and Gacon, Julien and Martiel, Simon and Nation, Paul D. and Bishop, Lev S. and Cross, Andrew W. and Johnson, Blake R. and Gambetta, Jay M.},
      year={2024},
      doi={10.48550/arXiv.2405.08810},
      eprint={2405.08810},
      archivePrefix={arXiv},
      primaryClass={quant-ph}
}

@article{jin2023quantum,
  title={Quantum simulation of partial differential equations: Applications and detailed analysis},
  author={Jin, Shi and Liu, Nana and Yu, Yue},
  journal={Physical Review A},
  volume={108},
  number={3},
  pages={032603},
  year={2023},
  publisher={APS}
}

@article{low2019hamiltonian,
  title={Hamiltonian simulation by qubitization},
  author={Low, Guang Hao and Chuang, Isaac L},
  journal={Quantum},
  volume={3},
  pages={163},
  year={2019},
  publisher={Verein zur F{\"o}rderung des Open Access Publizierens in den Quantenwissenschaften}
}

@article{lubasch2025quantum,
  title={Quantum circuits for partial differential equations in Fourier space},
  author={Lubasch, Michael and Kikuchi, Yuta and Wright, Lewis and Mc Keever, Conor},
  journal={Physical Review Research},
  volume={7},
  number={4},
  pages={043326},
  year={2025},
  publisher={APS}
}

@article{childs2021high,
  title={High-precision quantum algorithms for partial differential equations},
  author={Childs, Andrew M and Liu, Jin-Peng and Ostrander, Aaron},
  journal={Quantum},
  volume={5},
  pages={574},
  year={2021},
  publisher={Verein zur F{\"o}rderung des Open Access Publizierens in den Quantenwissenschaften}
}

@article{bungartz2004sparse,
  title={Sparse grids},
  author={Bungartz, Hans-Joachim and Griebel, Michael},
  journal={Acta numerica},
  volume={13},
  pages={147--269},
  year={2004},
  publisher={Cambridge University Press}
}

@article{de2024numerical,
  title={Numerical analysis of physics-informed neural networks and related models in physics-informed machine learning},
  author={De Ryck, Tim and Mishra, Siddhartha},
  journal={Acta Numerica},
  volume={33},
  pages={633--713},
  year={2024},
  publisher={Cambridge University Press}
}

@article{furuya2024quantitative,
  title={Quantitative approximation for neural operators in nonlinear parabolic equations},
  author={Furuya, Takashi and Taniguchi, Koichi and Okuda, Satoshi},
  journal={arXiv preprint arXiv:2410.02151},
  year={2024}
}

@article{cramer2010efficient,
  title={Efficient quantum state tomography},
  author={Cramer, Marcus and Plenio, Martin B and Flammia, Steven T and Somma, Rolando and Gross, David and Bartlett, Stephen D and Landon-Cardinal, Olivier and Poulin, David and Liu, Yi-Kai},
  journal={Nature communications},
  volume={1},
  number={1},
  pages={149},
  year={2010},
  publisher={Nature Publishing Group UK London}
}

@article{tong2021fast,
  title={Fast inversion, preconditioned quantum linear system solvers, fast Green's-function computation, and fast evaluation of matrix functions},
  author={Tong, Yu and An, Dong and Wiebe, Nathan and Lin, Lin},
  journal={Physical Review A},
  volume={104},
  number={3},
  pages={032422},
  year={2021},
  publisher={APS}
}

@article{kerenidis,
  title={Quantum recommendation systems},
  author={Kerenidis, Iordanis and Prakash, Anupam},
  journal={arXiv preprint arXiv:1603.08675},
  year={2016}
}

@article{preskill,
  title={Lecture notes for physics 229: Quantum information and computation},
  author={Preskill, John},
  journal={California institute of technology},
  volume={16},
  number={1},
  pages={1--8},
  year={1998}
}

@article{wavelet1,
  title={Efficient quantum algorithm for all quantum wavelet transforms},
  author={Bagherimehrab, Mohsen and Aspuru-Guzik, Al{\'a}n},
  journal={Quantum Science and Technology},
  volume={9},
  number={3},
  pages={035010},
  year={2024},
  publisher={IOP Publishing}
}

@article{wavelet2,
  title={Wavelet-based spectral analysis},
  author={Hoang, Vu Dang},
  journal={TrAC Trends in Analytical Chemistry},
  volume={62},
  pages={144--153},
  year={2014},
  publisher={Elsevier}
}

@article{antonioli2026quantumblockencodingsemiseparable,
  title={Quantum block encoding for semiseparable matrices}, 
   author={Giacomo Antonioli and Paola Boito and Gianna M. Del Corso and Margherita Porcelli},
  journal={arXiv preprint arXiv:2603.19130},
  year={2026}
}

@misc{motamedi2024group,
  author       = {Arsalan Motamedi and Grecia Castelazo},
  title        = {Group convolution quantum algorithms: an application to PDEs},
  howpublished = {Talk at QTML 2024 Conference, University of Melbourne},
  year         = {2024},
  month        = nov,
  note         = {24--29 November 2024},
  url          = {https://indico.qtml2024.org/event/1/contributions/214/}
}

@phdthesis{castelazo2022prospects,
  title={Prospects for Quantum Equivariant Neural Networks},
  author={Castelazo, Grecia},
  year={2022},
  school={Massachusetts Institute of Technology}
}
\end{document}